\documentclass[british]{scrartcl}
\usepackage[T1]{fontenc}
\usepackage[latin9]{inputenc}
\usepackage{float}
\usepackage{bm}
\usepackage{amsmath}
\usepackage{accents}
\usepackage{amssymb}
\usepackage{graphicx}
\usepackage{pdfpages}
\usepackage{esint}
\usepackage[authoryear]{natbib}
\usepackage{amsthm}
\usepackage{multirow}
\usepackage{url}
\usepackage{color}
\usepackage{enumitem}
\usepackage{hyperref}
\usepackage{caption}
\usepackage{tikz}
\usepackage{pdfsync}
\usepackage{subcaption}
\usepackage{todonotes}
\usetikzlibrary{positioning}
\usetikzlibrary{patterns}
\usetikzlibrary{shapes}
\usetikzlibrary{plotmarks}

\makeatletter

\floatstyle{ruled}
\newfloat{algorithm}{tbp}{loa}
\providecommand{\algorithmname}{Algorithm}
\floatname{algorithm}{\protect\algorithmname}

\usepackage{dsfont}\usepackage{algpseudocode}
\usepackage{mathtools}

\newcounter{rmq}[section]
\newcommand{\R}{\mathbb{R}}

\renewcommand{\P}{\mathbb{P}}
\newcommand{\E}{\mathbb{E}}
\newcommand{\var}{\mathrm{Var}}
\newcommand{\cov}{\mathrm{Cov}}

\newcommand{\ind}{\mathds{1}}
\newcommand{\dd}{\mathrm{d}}

\newcommand{\du}{\dd \mathrm{u}}

\newcommand{\eqdef}{:=} 
\newcommand{\bigO}{\mathcal{O}} 

\newcommand{\comment}[1]{ \ifthenelse{ \equal{\showcomment}{true} }{ {\bf #1} }{} }
\newcommand{\showcomment}{true}

\newtheorem{thm}{Theorem}
\newtheorem{prop}{Proposition}

\newtheorem{lemma}{Lemma}

\newtheorem{remark}{Remark}

\theoremstyle{plain}

\theoremstyle{plain}

\theoremstyle{plain}

\newcommand{\Ih}{\widehat{\mathcal{I}}} 
\newcommand{\Ihrk}[1]{\widehat{\mathcal{I}}_{#1,k}} 
\newcommand{\Itrk}[1]{\widetilde{\mathcal{I}}_{#1,k}}
\newcommand{\Ihork}[1]{\widehat{\mathcal{I}}^0_{#1,k}}

\newcommand{\Uk}{\mathcal{U}\left(\left[-\frac{1}{2k}, \frac{1}{2k}\right]^s\right)}
\newcommand{\Cx}[1]{\mathfrak{C}_{#1}}  
\newcommand{\Ck}{\Cx{k}}
\newcommand{\Cmk}{\Cx{m, k}}
\newcommand{\Cmrk}{\Cx{m_r, k}}
\newcommand{\Cinfk}{\mathfrak{C}_{\infty,k}}
\newcommand{\Cf}[1]{\mathcal{C}^{#1}\left([0, 1]^s\right)} 
\newcommand{\Cfo}[1]{\mathcal{C}^{#1}_0\left([0, 1]^s\right)} 

\let\oldReturn\Return
\renewcommand{\Return}{\State\oldReturn}

\newcommand{\reqstart}{
    \begin{list}{\thereqcount}{\usecounter{reqcount}}
    \setcounter{reqcount}{\value{reqcountbackup}}
}

\newcommand{\reqend}{
    \setcounter{reqcountbackup}{\value{reqcount}}
    \end{list}
}

  \newcommand\iid{\stackrel{\mathclap{\normalfont\mbox{\tiny{iid}}}}{\sim}}
   \newcommand\dist{\stackrel{\mathclap{\normalfont\mbox{\tiny{dist}}}}{=}}

\makeatletter
\newcommand*{\vneq}{%
  \mathrel{%
    \mathpalette\@vneq{=}%
  }%
}
\newcommand*{\@vneq}[2]{%
  \sbox0{\raisebox{\depth}{$#1\neq$}}%
  \sbox2{\raisebox{\depth}{$#1|\m@th$}}%
  \ifdim\ht2>\ht0 %
    \sbox2{\resizebox{\vneqxscale\width}{\vneqyscale\ht0}{\unhbox2}}%
  \fi
  \sbox2{$\m@th#1\vcenter{\copy2}$}%
  \ooalign{%
    \hfil\phantom{\copy2}\hfil\cr
    \hfil$#1#2\m@th$\hfil\cr
    \hfil\copy2\hfil\cr
  }%
}
\newcommand*{\vneqxscale}{1}
\newcommand*{\vneqyscale}{1}

\makeatother

\usepackage{babel}

\author{Nicolas Chopin$^{(1)}$, Mathieu Gerber$^{(2)}$
\\
\small{(1) ENSAE, Institut Polytechnique de Paris, Paris, France}
\\
\small{(2) School of Mathematics, University of Bristol, UK }
}

\begin{document}

\date{}
\title{Higher-order Monte Carlo through cubic stratification}

\maketitle

\date{}
\begin{abstract}
We propose two novel unbiased estimators of the integral
$\int_{[0,1]^{s}}f(u)\du$ for a function $f$, which depend on a smoothness
parameter $r\in\mathbb{N}$. The first estimator integrates exactly the
polynomials of degrees $p<r$ and  achieves the optimal error $n^{-1/2-r/s}$
(where $n$ is the number of evaluations of $f$) when $f$ is $r$ times
continuously differentiable. The second estimator  is also optimal in term of convergence rate and has the advantage to be computationally cheaper, but it is restricted to functions that vanish on the boundary of $[0,1]^s$. The construction of the two estimators relies on a combination of cubic stratification and control variates based on numerical derivatives.  We provide numerical evidence that they show good performance even for moderate values of $n$.
\end{abstract}

\section{Introduction\label{sec:intro}}

\subsection{Background}

This paper is concerned with the construction of unbiased estimators of the
integral $\mathcal{I}(f):=\int_{[0,1]^{s}}f(u)\du$ based on a certain number $n$
of evaluations of $f$. The motivation for this problem is well-known. Many
quantities of interest in applied mathematics may be expressed as such an
integral. Providing random, unbiased approximations present several practical
advantages. First, it greatly facilitates the assessment of the numerical error,
through repeated runs. Second, such independent estimates may be generated in
parallel, and then may be averaged to obtain a lower variance approximation of
$\mathcal{I}(f)$. Third, generating unbiased estimates as plug-in replacements
is of interest in various advanced Monte Carlo methodologies, such as
pseudo-marginal sampling \cite{Andrieu2009}, stochastic approximation
\cite{MR42668} and stochastic gradient descent. Finally, random integration
algorithms converge at a faster rate than deterministic ones  \cite{novak2006deterministic}
(but note that these convergence rates correspond to different criteria).

The most basic and well-known stochastic integration rule is the crude Monte
Carlo method, where one simulates uniformly $n$ independent and identically
distributed variates $U_{i}$, and returns $n^{-1}\sum_{i=1}^{n}f(U_{i})$ as an
estimate of $\mathcal{I}(f)$. Assuming that $f\in L_{2}([0,1]^s)$, the root mean
square error (RMSE) of this estimator converges to zero at rate $n^{-1/2}$. In
this paper we consider the problem of estimating $\mathcal{I}(f)$ under the
additional condition that all the partial derivatives of $f$ of order less or
equal to $r$ exist and are continuous, or, in short, that
$f\in\mathcal{C}^r([0,1]^s)$. Under this assumption on $f$ it is well-known that
we can improve upon the crude Monte Carlo error rate. More precisely, for
$f\in\mathcal{C}^r([0,1]^s)$ the optimal convergence rate for
the RMSE of an estimate $\Ih(f)$ of $\mathcal{I}(f)$ based on $n$
evaluations of $f$ is $n^{-1/2-r/s}$, in the sense that if
$g:\mathbb{N}\rightarrow[0,\infty)$ is such that 
\begin{equation*}
    \forall f\in \Cf{r}, n\geq 1,\quad \E\left[|\Ih(f)-\mathcal{I}(f)|^{2}\right]^{1/2}
\leq g(n) \|f\|_r
\end{equation*}
(where $\|f\|_r$ is a bound on the $r$-th order derivatives of $f$, see 
Section~\ref{sub:notTaylor} for a proper definition)
then we must have  $n^{-1/2-r/s}/g(n)=\bigO(1)$ (this result can for instance be obtained from Propositions 1-2 given in Section 2.2.4, page 55, of \cite{novak2006deterministic}).

Stochastic algorithms that achieve this optimal convergence rate for
$f\in\mathcal{C}^r([0,1]^s)$ have been proposed e.g.\ in
\cite {haber1966modified} for $r=1$ and in \cite{haber1967modified} for
$r\in\{1,2\}$. In \cite{haber1969stochastic} it is shown that if $\Ih(\cdot)$ is a
stochastic quadrature (SQ) of degree $r-1$, that is, if
$\E[\Ih(f)]=\mathcal{I}(f)$ for all $f\in L_{1}([0,1]^s)$ and
$\P(\Ih(f)=\mathcal{I}(f))=1$ if $f$ is a polynomial of degree $p<r$, then
$\Ih(\cdot)$ can be used to define an estimator of $\mathcal{I}(f)$ whose RMSE
converges to zero at rate $n^{-1/2-r/s}$ when $f\in\mathcal{C}^r([0,1]^s)$. In
\cite{haber1969stochastic} a formula for a SQ of degree $r-1$ is given for
$r\in\{3,4\}$ while, for $s=1$, \cite{siegel1985unbiased} provides a SQ of
degree $2r+1$ for all $r\geq 1$. For multivariate integration problems, and an
arbitrary value of $r\geq 1$, a SQ of degree $r-1$ can be constructed from the
integration method proposed in \cite{ermakov1960polynomial}. However, the
algorithm proposed in this reference requires to perform a sampling task which
is so computationally expensive that it is considered as almost intractable
\cite{patterson1987construction}.

A related approach is derived by Dick in \cite{dick2011higher}, which achieves rate
$\bigO(n^{-1/2-\alpha+\varepsilon})$ for $\varepsilon>0$ and a certain class of
functions indexed by $\alpha$ (which differs from $\Cf{r}$ even when $r=s\alpha$).
We will go back to this point and compare our approach to Dick's in our
numerical study.

\subsection{Motivation and plan}

The paper is structured as follows.
We introduce in Section~\ref{sec:main} an unbiased estimator
of $\mathcal{I}(f)$ which has the following three appealing
properties when $f\in\mathcal{C}^r([0,1]^s)$. First, its RMSE converges to zero
at the optimal $n^{-1/2-r/s}$ rate. Second, it integrates exactly $f$ if $f$ is
a polynomial of degree $p<r$. Third, for some constant $C<\infty$ and with
probability one, the absolute value of its estimation error is bounded by $C
n^{-r/s}$, where $n^{-r/s}$ is the optimal convergence rate for a deterministic
integration rule (this result can for instance be obtained from Proposition 1.3.5, page 28, of \cite{novak2006deterministic}). In addition, we establish a central limit theorem
(CLT) for a particular version of the proposed estimator. To the best of our
knowledge, a CLT for an estimator of $\mathcal{I}(f)$ having an RMSE that
converges at the optimal rate when $f\in\mathcal{C}^r([0,1]^s)$ exists only for
$r=1$ (see \cite{bardenet2020monte}).

In Section~\ref{sec:vanish}, we focus our attention on the estimation of
$\mathcal{I}(f)$ when $f\in\mathcal{C}^r_0([0,1]^s)$, where we define
$\mathcal{C}^r_0([0,1]^s)$ as the set of functions in $\mathcal{C}^r([0,1]^s)$
whose partial derivatives of order $o\leq r$ are all equal to zero on the
boundary of $[0,1]^s$. As we explain in that section, this set-up is
particularly relevant for solving integration problems on $\R^s$. Restricting
our attention to $\mathcal{C}^r_0([0,1]^s)\subset \mathcal{C}^r([0,1]^s)$
allows us to derive an estimator of $\mathcal{I}(f)$, referred to  as the
vanishing estimator in what follows, which is computationally cheaper  than the
previous estimator, while retaining its convergence properties, namely an RMSE
of size $\bigO(n^{-1/2-r/s})$  and an actual error of size $\bigO(n^{-r/s})$
almost surely. We note that these convergence rates   are optimal for
integrating a function in $\mathcal{C}^r_0([0,1]^s)$ (again, see Sections 1.3.5
and 2.2.4  of \cite{novak2006deterministic}) and that an algorithm considering
a similar class of functions is proposed in \cite{krieg2017universal}. The
algorithm derived in this latter reference has the advantage to achieve the
optimal aforementioned convergence rates for any $r\in\mathbb{N}$ but its
implementation at reasonable computational cost remains an open problem.

Section~\ref{sec:pract} discusses some practical details about the proposed
estimators, regarding on how their variance may be estimated  and how the order
of the vanishing estimator may be selected automatically. 
Section~\ref{sec:numerics} presents numerical experiments which confirm that
the estimators converge at the expected rates, and show  that they are already
practical for moderate values of $n$.  
Section~\ref{sec:conclusion} discusses future work. 
Proofs of certain technical lemmas are deferred to Appendix \ref{app-proofs}.

\subsection{Connection with function approximation}

As noted by e.g. \cite{novak2016some}, there is a strong connection between
(unbiased) integration and function approximation. If one is able to construct
an optimal approximation $\mathcal{A}_n(f)$ of $f\in\mathcal{C}^r{([0, 1]^s)}$,
that is $\|f-\mathcal{A}_n(f)\|_\infty=\bigO(n^{-r/s})$
(see \cite{novak2006deterministic}, page 36) then one may derive the following unbiased estimate of $\mathcal{I}(f)$
\begin{equation}
    \label{eq:estimate_func_approx}
  \Ih(f) :=
  \mathcal{I}\left(\mathcal{A}_n(f)
\right)+\frac{1}{n}\sum_{i=1}^n\left(f-\mathcal{A}_n(f)
\right)(U_i),\quad
  U_i\iid\mathcal{U}([0,1]^s)
\end{equation}
which is also optimal, in the sense that its RMSE is $\bigO(n^{-1/2-r/s})$ for estimating $\mathcal{I}(f)$.

The (non-vanishing) estimator  proposed in this paper for integrating a function
$f\in\mathcal{C}^r{([0, 1]^s)}$ is  to some extent  related to this idea, with
$\mathcal{A}_n(f)$ a piecewise polynomial approximation of $f$ based on local
Taylor expansions in which  the partial derivatives of $f$ are approximated
using numerical differentiation techniques. Note however that we use stratified
random variables, rather than independent and identically distributed ones.
This makes the estimator easier to compute,  and reduces its variance.

\subsection{Notation regarding derivatives and Taylor
expansions\label{sub:notTaylor}}

Let $\mathbb{N}$ be the set of positive integers, and
$\mathbb{N}_0=\mathbb{N}\cup\{0\}$. For $\alpha\in\mathbb{N}_0^s$, let
$|\alpha|_0=s-\sum_{j=1}^s\ind_{\{0\}}(\alpha_j)$,
$|\alpha|=\sum_{i=1}^s\alpha$, $\alpha!=\prod_{i=1}^s\alpha_i!$, $u^\alpha=
\prod_{i=1}^s u_i^{\alpha_i}$ for $u\in \R^s$. For
$g\in\mathcal{C}^{r}([0,1]^s)$ we let $D^\alpha g:[0,1]^s \rightarrow\R$ be
defined by
\begin{equation*}
D^\alpha g(u) = \frac{\partial^{|\alpha|}}{\partial u_1^{\alpha_1}\dots \partial u_s^{\alpha_s}}g(u),\quad u\in [0,1]^s,
\end{equation*}
with the convention $D^\alpha g=g$ if $|\alpha|=0$, and we let
$\|g\|_r\eqdef\max_{\alpha:\, |\alpha|=r}\|D^\alpha g\|_\infty$.

With this notation in place, we recall that if $g\in\mathcal{C}^{r}([0,1]^s)$
then, by Taylor's theorem, there exists a function
$R_{g,r}:[0,1]^s\times[0,1]^s\rightarrow\R$ such that (\cite{calculus}, Section 3.17,
page 191)
\begin{align}\label{eq:Taylor}
g(v)=\sum_{l=0}^{r-1}\sum_{\alpha:|\alpha|=l}(v-u)^\alpha\frac{D^\alpha g(u)}{\alpha!}+R_{g,r}(u,v),\quad\forall u,v\in[0,1]^s
\end{align}
where, for some $\tau_{u,v}\in[0,1]$,
\begin{align}\label{eq:Taylor2}
 R_{g,r}(u,v)=\sum_{\alpha:|\alpha|=r}\frac{D^\alpha g\left(u+\tau_{u,v}(v-u)
\right)}{\alpha !} \,  (v-u)^\alpha.
\end{align}

\subsection{Notation related to stratification}

Throughout the paper, $f:[0, 1]^s\rightarrow \R$ and $s\geq 1$. Our approach
relies on stratifying $[0, 1]^s$ into $k^s$ closed hyper-cubes, $k\geq 2$, and
performing a certain number $l$ of evaluations of $f$ inside each hyper-cube;
see Figure~\ref{fig:square}.
The total number of evaluations is therefore something like $n=lk^s$, but with a
value for $l$ that depends on the considered estimator and other parameters such
as $r$. Thus, we will index the proposed estimators by $k$,
e.g $\Ih_k(f)$ (or $\Ih_{r,k}(f)$ when it also depends on $r$) rather than $n$.
We will provide the exact expression of $n$ alongside the definition of the
considered estimator.

\begin{figure}[ht]
  \centering
  \includegraphics[scale=0.6]{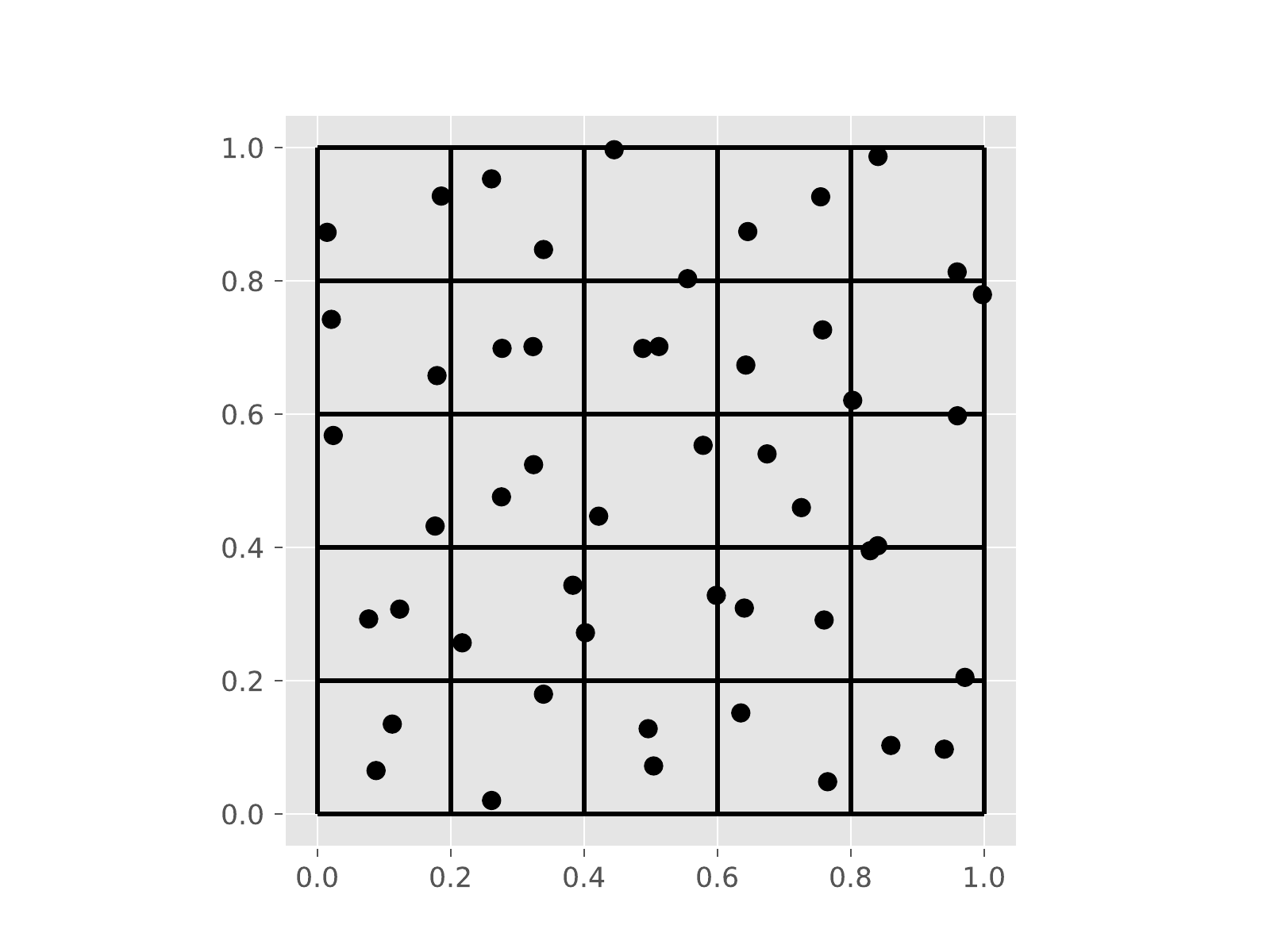}
  \caption{Stratification of $[0, 1]^s$ when $s=2$ and $k=5$, and two
    evaluations are performed in each of the $k^s=25$ squares. The location of
    the points are generated as in Haber's second estimator, which we discuss
    in Section~\ref{sub:haber}.
  \label{fig:square}}
\end{figure}

For $c\in\R^s$ and $k\geq 1$,  we use the short-hand $B_k(c)$ for the
hyper-cube  $\prod_{i=1}^s [c_i-1/2k, c_i+1/2k]$; in other words,
the ball with radius $1/2k$ and centre $r$ with respect to the maximum norm. 

For  $m\in\mathbb{N}_0$ let
\begin{equation}
    \label{eq:def_Cmk}
    \mathfrak{C}_{m,k}= \left\{\left( \frac{2j_1+1}{2k},\dots,
        \frac{2j_s+1}{2k}\right)\text{ s.t.
}\,(j_1,\dots,j_s)\in\{-m,\dots,k+m-1\}^s\right\}
\end{equation}
be the set of the centres of the $(k+2m)^s$ hypercubes $B_k(c)$ 
whose union is equal to the set $\mathcal{S}_{m,k}:=[-m/k,1+m/k]^s$.
In Section~\ref{sec:main}, we will set $m=0$ and recover the aforementioned
stratification; in that case, we will use the short-hand $\Ck\eqdef
\mathfrak{C}_{0,k}$.  However, in order to define the second (vanishing)
estimator in Section~\ref{sec:vanish}, we shall take $m\geq 0$.

To each $c\in\Cmk$ (with $m$, again, fixed and determined by the context), we
associate a random variable $U_c$ such that 
\begin{equation*}
    U_c \sim \mathcal{U}\left( \left[ - \frac{1}{2k}, \frac{1}{2k} \right]^s
    \right).
\end{equation*}
Notice that the support of $c+U_c$ is $B_k(c)$.

\section{Integration of functions in $\Cf{r}$\label{sec:main}}

\subsection{Preliminaries: Haber's estimators\label{sub:haber}}

In \cite{haber1966modified} Haber introduced the following estimator:
\begin{equation}
    \label{eq:Haber}
    \Ihrk{1}(f) \eqdef \frac{1}{k^s} \sum_{c\in\Ck} f(c+U_c),
    \qquad U_c \sim \Uk
\end{equation}
based on $n=k^s$ evaluations of $f$, which is optimal for $r=1$; i.e. its RMSE is
$\bigO(n^{-1/2-1/s})$ provided $f\in\mathcal{C}^1([0, 1]^s)$. To establish this
result, note that each term $f(c+U_c)$ has expectation $k^s\int_{B_k(c)} f(u)\dd u$ 
and variance $\bigO(n^{-2/s})$, since $|f(u)-f(v)| =
\bigO(k^{-1})=\bigO(n^{-1/s})$ for $u, v\in B_k(c)$. 

We note in passing that an alternative, and closely related, estimator may be
obtained by approximating $f$ with the   piecewise constant function $f_n$ defined by
\begin{equation*}
    f_n(u) = \sum_{c\in\Ck} f(c) \ind_{B_k(c)}(u),\quad u\in [0,1]^s
\end{equation*}
and using that particular $f_n$ in \eqref{eq:estimate_func_approx}. Since 
$\|f-f_n\|_\infty = \bigO(n^{-1/s})$ when $f\in\mathcal{C}^1([0,1]^s)$, this 
alternative estimator is indeed optimal for $r=1$. The estimator defined in \eqref{eq:Haber} is however slightly more convenient to compute, and relies on only $n$ evaluations (versus
$2n$ for the alternative estimator).

In \cite{haber1967modified} Haber introduced a second estimator:
\begin{equation}
    \label{eq:HaberII}
    \Ihrk{2}(f) \eqdef \frac{1}{k^s} \sum_{c\in\Ck} g_c(U_c),
    \qquad U_c \sim \Uk
\end{equation}
with $g_c(u) \eqdef \{f(c+u) + f(c-u)\} /2$ and $n=2k^s$,  which is optimal for
$f\in\mathcal{C}^2([0,1]^s)$. Note that $g_c$ is a symmetric function, thus 
its Taylor expansion at $0$ only includes even order terms: 
\begin{equation}
    \label{eq:Taylor_gc}
    g_c(u) = f(c) + \frac 1 2 u^T H_f(c) u + \bigO(k^{-4}),
    \quad \mbox{for } u\in B_k(c)
\end{equation}
where $H_f(c)$ denotes the Hessian matrix of $f$ at $c$. The term 
$g_c(U_c)$ has variance $\bigO(n^{-4/s})$ when $f\in\Cf{2}$, leading to an
$\bigO(n^{-1/2-2/s})$ RMSE for $\Ihrk{2}(f)$.

The estimators introduced in this paper have the same form as Haber's two
estimators; i.e. an average of terms $g_{r,c}(U_c)$, where $g_{r,c}(U_c) =
f(c) + \bigO(k^{-r})$ essentially.  To achieve this, we consider two
approaches: one based on control variates (this section), and another based on
combining more than two terms of the form $f(c+\lambda U_c)$ (Section
\ref{sec:vanish}).

\subsection{Control variates}


One simple way to improve on Haber's second estimator is to add a control
variate based on a Taylor expansion of $g_c$. To fix ideas, suppose that
$f\in\Cf{4}$, and add to each term $g_c(U_c)$ in~\eqref{eq:HaberII} the quantity 
\begin{equation*} 
- \frac 1 2 U_c^T H_f(c) U_c + \E\left[ \frac 1 2 U_c^T H_f(c) U_c \right]. 
\end{equation*}
This does not change
the overall expectation, since this extra term has zero mean, and, 
given~\eqref{eq:Taylor_gc}, it reduces the variance of each term to
$\bigO(n^{-8/s})$. 

More generally, for $r\geq 2$, let  $p_{c,r-1}$ be the polynomial function that
corresponds to the $(r-1)$-order Taylor expansion of $g_c$ at $0$, i.e. 
\eqref{eq:Taylor} with $g=g_c$ and $u=0$. Then, using \eqref{eq:Taylor} and
\eqref{eq:Taylor2}, we have 
\begin{align*}
    |g_c(u)-p_{c,r-1}(u)|\leq C\|f\|_r\|u\|^r,\quad\forall u\in 
    [1/2k,1-1/2k]^s,\quad\forall c\in\Ck
\end{align*} 
for some constant $C<\infty$ (which does not depend on $c$).  Letting
\begin{equation*}
    V_{r,k}(f) \eqdef
    -\frac{1}{k^s}\sum_{c\in\Ck}
    \left\{p_{c,r-1}(U_c)-\E[p_{c,r-1}(U_c)]\right\},
\end{equation*}
the variance of the estimator
$\mathcal{I}^{*}_{r, k}(f):=\Ihrk{2}(f)+V_{r,k}(f)$
is therefore such that
 \begin{align*}
   \var\left[  \mathcal{I}^{*}_k(f)\right]
  & =\frac{1}{k^{2s}}\sum_{c\in\Ck} \var\left[  g_c(U_c)-p_{c,r-1}(U_c)\right]\\
  & \leq  C^2 \|f\|_r^2 \times k^{-s-2r}\\
  & = C'\,\|f\|^2_r\,n^{-1-2r/s}, \qquad C'\eqdef C^2\times 2^{1+2r/s}.
\end{align*}

Since $\mathcal{I}^{*}_k(f)$ is an unbiased estimator of $\mathcal{I}(f)$,
its RMSE is $\bigO(n^{-1/2-r/s})$. Moreover, with probability one
$\mathcal{I}^{*}_k(f)=\mathcal{I}(f)$ if $f$ is a polynomial of degree $p<r$
since, in this case, $\|f\|_r=0$.

The main drawback of  estimator $\mathcal{I}^*_k(f)$ is that it requires to
compute and evaluate derivatives of $f$; that may be feasible in certain cases
(using for instance automatic differentiation, see \cite{MR3800512}). However,
it is generally simpler to have an estimator that relies only on evaluations of
$f$.  Surprisingly, and as shown in the following, higher-order difference
methods make it possible to replace, in the definition of $p_{c,r-1}$, the
partial derivatives of $f$ by numerical derivatives while preserving the
convergence rate of $\mathcal{I}^*_k(f)$  as well as its ability to integrate
exactly polynomials of degree $p<r$.


Higher-order difference methods are widely used in practice for numerical
differentiation. However it is surprisingly hard
to find a reference providing an explicit definition of an estimate
$\hat{D}^\alpha f$ of $D^\alpha f$ along with an explicit error bound
$e_f(s,r,|\alpha|)$ for the approximation error $\|\hat{D}^\alpha f-D^\alpha
f\|_\infty$. For this reason, in the next subsection we provide two results on
numerical differentiation based on higher-order difference methods before coming
back to the estimation of $\mathcal{I}(f)$ in the subsequent subsections.

\subsection{Numerical differentiation\label{sub:num_diff}}

The result in the following lemma can be used, for $s=1$, to compute an estimate
$\hat{D}^\alpha f$ of $D^\alpha f $ as well as to obtain an upper bound for
the approximation error. 

\begin{lemma}\label{lemma:num_deriv}
Let $g\in\mathcal{C}^{l}([0,1])$ for some integer $l\geq 2$,
$a\in\{1,\dots,l-1\}$ and $\kappa\in\R^l$ be a vector containing $l$ distinct
elements. Next, let $e^{(a)}\in\R^l$ be such that $e_{a+1}^{(a)}=a!$, $e_j^{(a)}=0$
for $j\neq (a+1)$, and let   
\begin{equation*}
w = A_\kappa^{-1} e^{(a)},
\quad A_\kappa=
\begin{pmatrix}
1&1&\hdots&1\\
\kappa_1&\kappa_2&\hdots&\kappa_{l}\\
\vdots&\vdots&\vdots&\vdots\\
\kappa_1^{l-1}&\kappa_2^{l-1}&\hdots&\kappa_{l}^{l-1}
\end{pmatrix}.
\end{equation*}
Let $x\in [0,1]$ and $h>0$ be such that $x+\kappa_j h\in(0,1)$ for all $j\in\{1,\dots,l\}$. Then,
\[
\left|\frac{\sum_{j=1}^l w_j g(x+\kappa_j h)}{h^{a}}-g^{(a)}(x)\right|
\leq h^{l-a} \|g\|_l\sum_{j=1}^l |w_j\kappa_j^l|.
\]
\end{lemma}

\begin{remark}
The matrix $A_\kappa$ is invertible since $A_\kappa$ is a Vandermonde matrix
and $\kappa_j\neq \kappa_l$ for all $j\neq l$.
\end{remark}

\begin{proof}
By construction, $\{w_j\}_{j=1}^{l}$ is such that $\sum_{j=1}^{l}
w_j\kappa^i_j=0$ for all $i\in\{0,\dots,l-1\}\setminus\{a\}$ and such that
$\sum_{j=1}^l w_j\kappa^a_j=a!$. Therefore, using
\eqref{eq:Taylor}-\eqref{eq:Taylor2}, for some $\{\tau_j\}_{j=1}^l$ in $[-1,1]$
we have
\begin{align*}
\sum_{j=1}^l w_j g(x+\kappa_j h)
& =\sum_{i=0}^{l-1}h^i\frac{g^{(i)}(x)}{i!}\left(\sum_{j=1}^l w_j \kappa_j^i\right)
    +\sum_{j=1}^l w_j (\kappa_j h)^l g^{(l)}(x+\tau_j \kappa_j h)\\
& = h^a g^{(a)}(x)+h^l\left(\sum_{j=1}^l w_j\kappa^l_j\, g^{(l)}(x+\tau_j \kappa_j h)\right)
\end{align*}
and thus
\begin{align*}
  \left|g^{(a)}(x)-\frac{\sum_{j=1}^l w_j g(x+\kappa_j h)}{h^a}\right|
  &\leq h^{l-a}\left|\sum_{j=1}^l w_j\kappa_j^l\, g^{(l)}(x+\tau_j \kappa_j h)\right|\\
  & \leq h^{l-a}\|g^{(l)}\|_\infty \sum_{j=1}^l |w_j\kappa_j^l |.
\end{align*}
The proof is complete.
\end{proof}

\begin{remark}
Usually, one sets the $\kappa_j$'s to small integers; e.g.
$\kappa=(0, 1, 2)$ for $l=3$ and $a=2$ gives the well-known forward formula with
first-order accuracy: 
\[ \frac{g(x) - 2 g(x+h) + g(x + 2h)}{h^2} = g^{(2)}(x) + \bigO(h).\]
If one uses instead so-called central coefficients, e.g. $\kappa=(-1, 0, 1)$,
then one may actually get an extra order of accuracy:
\[ \frac{g(x - h) - 2 g(x) + g(x + h)}{h^2} = g^{(2)}(x) + \bigO(h^2)\] as one
can check from first principles. We stick to the general case to keep our
notations simpler, as it will not have an impact on our general results.
\end{remark}

In our case, we need to compute (multivariate) numerical derivatives of $f$ at
all the points $c\in\Ck$, simultaneously.  The previous lemma indicates that a
numerical derivative of $f$ at $c$ is a linear combination of terms of the form
$f(c+\kappa_j  h)$. If we take $h=1/k$, and $\kappa_j\in\mathbb{N}$, then
$(c+\kappa_j h) \in \Ck$ (unless $c+\kappa_j\not\in [0, 1]^s$, which should
happen if $c$ is too close to the boundary). This suggests the following
strategy: first, compute $f(c)$ for all $c\in \Ck$; then, for a given
$\alpha\in\mathbb{N}_0^s$, approximate $D^\alpha f$ at each $c\in\Ck$ by
computing appropriate linear combinations of these $f(c)$. 

The following lemma formalises this remark. We note in passing that this trick
seems to be already known; it is implemented for instance in the package
findiff of
\cite{findiff} (which we use in our numerical experiments, see
Section~\ref{sec:numerics}), although it seems rarely mentioned in books on
numerical analysis.

 \begin{lemma}\label{lemma:num_deriv2}
Let $r\geq 2$.  Then,  there exist  finite  constants $\{C_{i,s}\}_{i=1}^{r-1}$
and a finite set $\mathcal{W}_{r}$ of real numbers,  which does not depend on
$s$, for which the following holds. For $k\geq r$,   $c\in\Ck$,  $\alpha$ such
that $|\alpha|<r$ and  $l_{r,\alpha}:=\prod_{i=1}^{|\alpha|_0} (r-i+1)$
elements $\{c^{(q)}\}_{q=1}^{l_{r,\alpha}}$ of $\Ck$ such that
\begin{enumerate}
\item $\|c-c^{(q)}\|\leq (r-1)/k$ for all $q\in\{1,\dots,l_{s,\alpha}\}$,
\item for all $j\in\{1,\dots,s\}$, if $\alpha_j=0$ then $c^{(q)}_j= c_j$  for all $q\in\{1,\dots,l_{s,\alpha}\}$,
\item for all $j\in\{1,\dots,s\}$, if $\alpha_j\neq 0$ then $c_j^{(q)}\neq  c_j^{(q')}$ for all $q, q'\in\{1,\dots,l_{s,\alpha}\}$ such that $q\neq q'$,
\end{enumerate}
there exist real numbers $\{w_{j}\}_{j=1}^{l_{r,\alpha}}$ such that
\begin{itemize}
\item each $w_j$  is the product of  $|\alpha|_0$ elements of the set $\mathcal{W}_{r}$,
\item the set $\{w_j\}_{j=1}^{l_{r,\alpha}}$ depends on $c$ only through the set $\{ k(c^{(q)}-c)\}_{q=1}^{l_{r,\alpha}}$, and is therefore independent of $k$,
\item for all $f\in\mathcal{C}^r([0,1]^s)$ we have
$
\big|\widehat{D}^{\alpha}_{k}f(c)-D^{\alpha} f(c)\big|\leq C_{|\alpha|,s}\, \|f\|_r\, k^{-(r-|\alpha|)}
$, where 
\begin{align}\label{eq:deriv}
\widehat{D}^{\alpha}_{k}f(c)=k^{|\alpha|}\sum_{j=1}^{l_{r,\alpha}} w_{j} f(c^{(j)}).
\end{align}
\end{itemize}
\end{lemma}

For what follows it is important to stress that, in~\eqref{eq:deriv}, the sets
$\{w_{j}\}_{j=1}^{l_{r,\alpha}}$ and   $\{c^{(j)}\}_{j=1}^{l_{r,\alpha}}$ are
independent of $f$  and
that the computational cost of computing these two sets is independent of $k$. We
also point out that, building on Lemma \ref{lemma:num_deriv}, the proof of
Lemma \ref{lemma:num_deriv2} is constructive and thus can be used in practice
to compute a numerical derivative $\widehat{D}^{\alpha}_{k}f(c)$ as defined in \eqref{eq:deriv}.

\subsection{Proposed estimator\label{sub:hat_I}}

Let $r\geq 3$,    $k\geq r$ and $f:[0, 1]^s \rightarrow \R$. Then, the  proposed estimator
 of $\mathcal{I}(f)$ is  
\begin{multline}
    \label{eq:estimator}
    \Ihrk{r}(f)  \eqdef \frac{1}{k^s} \sum_{c\in\Ck}
    \Bigg\{\frac{f(c+U_c)+f(c-U_c)}{2} \\
- \sum_{l=1}^{\lfloor (r-1)/2\rfloor}\sum_{\alpha:\,
|\alpha|=2l}\frac{\widehat{D}^{\alpha}_{k}f(c)}{\alpha!}
\left( U^\alpha_{c}- \prod_{j=1}^s d_{k}(\alpha_j)\right) \Bigg\}
\end{multline}
where the numerical derivatives $\widehat{D}^{\alpha}_{k}f(c)$'s are  as
in Lemma \ref{lemma:num_deriv2}  and 
\begin{align*}
d_{k}(i)\eqdef\E[V^i] = 
\begin{cases}
    \frac{1}{(i+1)(2k)^i} & \mbox{if $i$ is even,} \\
    0 & \mbox{otherwise,}
\end{cases}
\qquad\mbox{with }V\sim\mathcal{U}\left(\left[-\frac{1}{2k},
\frac{1}{2k}\right]\right).
\end{align*}

This estimator is based on $n=3k^s$ evaluations of $f$: two thirds at random
locations, and one third at deterministic locations (the $f(c)$'s for $c\in\Ck$
which are used to compute the derivatives).  Note that
$\widehat{\mathcal{I}}_{2q,k}(f)=\widehat{\mathcal{I}}_{2q-1,k}(f)$ for all
$q\geq 1$, and that only even-order derivatives appear in~\eqref{eq:estimator}, because
of the symmetry of function $g_c(u) = \{f(c+u) + f(c-u)\}/2$ (as 
explained before). 

The main drawback of the estimator $\widehat{\mathcal{I}}_{r,k}(f)$ is that its
computational cost increases quickly with $r$ and $s$.  We may rewrite the
second term of~\eqref{eq:estimator} as
\begin{equation*}
    \sum_{c\in\Ck} W_{r, c} f(c)
\end{equation*}
where the $W_{r,c}$'s are random weights that do not depend on $f$.  The 
number of partial derivatives of $f$ of order $|\alpha|$ being equal to 
$\binom{s+|\alpha|-1}{s-1}$, the number of operations required  to compute
these weights is:
\begin{equation*}
\bigO\left(s r^2 k^s
 \sum_{l=1}^{\lfloor (r-1)/2\rfloor} \binom{s+2l-1}{s-1}\right)
=\bigO( r^{s+3} k^s)
\end{equation*}
which is exponential in $s$.

On the other hand, since the $W_{r, c}$'s are independent of $f$, they may be
pre-computed, and re-used for several functions $f$. Alternatively, when $f$ is
expensive to compute, the cost of computing these $W_{r, c}$
will remain negligible (relative to the cost of the $n$ evaluations of $f$)
whenever $s$ and $r$ are not too high. See Section \ref{sub:Pima1}  for a
practical example where the function of interest $f$ is expensive to compute.

\subsection{An alternative estimator}

The estimator $\widehat{\mathcal{I}}_{r,k}(f)$ defined in \eqref{eq:estimator}
was obtained by 
adding control variates  to Haber's second estimator~\eqref{eq:HaberII}. 
By adding similar variates to his 
first estimator~\eqref{eq:Haber}, we obtain 
the following alternative estimator: 

\begin{equation}\label{eq:estimator_tilde}
\begin{split}
\widetilde{\mathcal{I}}_{r,k}(f)&:=\frac{1}{k^s}\sum_{c\in\Ck}\left(
f(c+U_c)-\sum_{l=1}^{r-1}\sum_{\alpha:\,
|\alpha|=l}\frac{\widehat{D}^{\alpha}_{k}f(c)}{\alpha!}\left(U_{c}^\alpha-
\prod_{j=1}^s d_{k}(\alpha_j)\right)\right)
\end{split}
\end{equation}
with the  derivatives $\widehat{D}^{\alpha}_{k}f(c)$'s   as in 
Lemma~\ref{lemma:num_deriv2}. 

The estimator $\Itrk{r}(f)$ has  the advantage of requiring only $n=2k^s$
evaluations of $f$, against $n=3k^s$ for  $\Ihrk{r}(f)$. In addition,
$\Itrk{r}(f)$ has a different expression for each value of $r$, while
$\Ihrk{2q}(f)=\Ihrk{2q-1}(f)$.

On the other hand, computing  $\Itrk{r}(f)$ is more expensive than
$\Ihrk{r}(f)$, since the former requires  to approximate all the partial
derivatives of $f$ of  order $|\alpha|<r$, while the latter necessitates  to
compute  only those having  an even order. 

In our numerical experiments, we implement only $\Ihrk{r}(f)$. But, for the
sake of completeness, we shall state the properties of both estimators in the
following section. 


\subsection{Error bounds\label{sub:bound}}

The error bounds presented in this subsection follow directly from the
following key lemma, whose proof is in   Appendix \ref{p-lemma:main}.

\begin{lemma}\label{lemma:main}
Let $f\in\mathcal{C}^r([0,1]^s)$ for some $r\geq 1$. Then 
there exists, for each $c\in\Ck$ and $k\geq r$, a function $h_{k,c}:[-1/2k,
1/2k]^s \rightarrow \R$ (which depends  implicitly on $r$) 
such that
\begin{equation*}
    \Ihrk{r}(f) - \mathcal{I}(f) = \frac{1}{k^s}\sum_{c\in\Ck} h_{k, c}(U_c)
\end{equation*}
and such that, for a constant  $\widehat{C}_{s,r}<\infty$ independent of $k$ and $f$,
\begin{equation*}
    \max_{c\in\Ck} \| h_{k, c}\|_\infty \leq \widehat{C}_{s,r} \|f\|_r k^{-r}.
\end{equation*}
This statement also holds for $\Itrk{r}(f)$.
\end{lemma}

The following theorem provides three types of error bounds for the estimators   $\widehat{\mathcal{I}}_{r,k}(f)$ and $\widetilde{\mathcal{I}}_{r,k}(f)$, namely  an error bound for the RMSE, an error bound that holds with  probability one and an error bound that holds with large probability. We recall that the number $n$ of evaluations  of $f$ is $n=3 k^s$ for the former estimator and $n=2k^s$ for the latter.

\begin{thm}\label{thm:main}
Let $f\in\mathcal{C}^r([0,1]^s)$ for some $r\geq 1$ and let
$\widehat{C}_{s,r}<\infty$ be as in Lemma \ref{lemma:main}. Then, for all   $k\geq r$,
\begin{enumerate}
\item\label{res1}  $\E\left[\Ihrk{r}(f)\right]=\mathcal{I}(f)$,
\item\label{res2}
    $\left[ \E \big|\Ihrk{r}(f)-\mathcal{I}(f)\big|^2\right]^{\frac{1}{2}}\leq 
    \widehat{C}_{s,r}\, \|f\|_r\, n^{-\frac{1}{2}-\frac{r}{s}}$,
\item $\P\left(\big|\widehat{\mathcal{I}}_{r,k}(f)-\mathcal{I}(f)\big|\leq
        \widehat{C}_{s,r}\, \|f\|_r\, n^{-\frac{r}{s}}\right) = 1$,
\item\label{res4} For all $\delta\in (0,1)$,  
\begin{equation*}
\P\left\{ \big|\widehat{\mathcal{I}}_{r,k}(f)-\mathcal{I}(f)\big|\leq
n^{-\frac{1}{2}-\frac{r}{s}}\,\widehat{C}_{s,r}\|f\|_r  
\sqrt{2 \log(2 /\delta)} \right\}\geq 1-\delta.
\end{equation*}
\end{enumerate}
The results given in \ref{res1}-\ref{res4} also hold with
$\widehat{\mathcal{I}}_{r,k}(f)$   replaced by
$\widetilde{\mathcal{I}}_{r,k}(f)$.
\end{thm}
\begin{proof}
We have already mentioned that $\Ihrk{r}(f) = \Ihrk{2}(f) + \widehat{V}_{r,
k}(f)$, where $\Ihrk{2}(f)$ is unbiased \cite{haber1967modified}  and 
$\widehat{V}_{r, k}(f)$ has zero mean. (The same remarks apply to $\Itrk{k}(f)$.)
The second and third parts of the
theorem are direct consequences of Lemma \ref{lemma:main}  and the last part of
the theorem is a direct consequence of Lemma \ref{lemma:main} and of
Hoeffding's inequality.
\end{proof}

The second part of the theorem shows that $\widehat{\mathcal{I}}_{r,k}(f)$
and $\widetilde{\mathcal{I}}_{r,k}(f)$ are optimal  for integrating a function
$f\in\mathcal{C}^r([0,1]^s)$, in the sense that their RMSE converge to zero
at the optimal rate (see Section~\ref{sec:intro}). The third part of the
theorem states that each realization of the estimators achieves the optimal
convergence rate for a  deterministic algorithm (again, see Section
\ref{sec:intro}). The last part of the theorem shows that  the distribution of 
$\widehat{\mathcal{I}}_{r,k}(f)$ and  of $\widetilde{\mathcal{I}}_{r,k}(f)$ are
sub-Gaussian. Finally, and importantly,  Theorem \ref{thm:main} shows that for
any $k\geq r$ the estimators $\widehat{I}_{r,k}(f)$ and
$\widetilde{I}_{r,k}(f)$ are exact if $f$ is a polynomial of degree $p<r$.
Indeed,  if $f\in\mathcal{C}^r([0,1]^s)$ is a polynomial of degree $p<r$ then  
$\|f\|_r=0$ and thus, by Theorem \ref{thm:main},
\begin{equation*}
\P\left(\widehat{\mathcal{I}}_{r,k}(f)=\mathcal{I}(f)
\right)=1,\quad \P\left(\widetilde{\mathcal{I}}_{r,k}(f)=\mathcal{I}(f)
\right)=1,\quad\forall k\geq r.
\end{equation*}

\subsection{Central limit theorem}

The following lemma provides a sufficient condition for a central limit theorem
to hold for $\Ihrk{r}(f)$ and $\Itrk{r}(f)$.

\begin{lemma}\label{lemma:CLT}
Let $f\in\mathcal{C}^r([0,1]^s)$ for some $r\geq 1$ and assume that there
exists a sequence $(v_k)_{k\geq 1}$ such that $v_k\rightarrow\infty$ and such that
\begin{equation*}
\var\left(\widehat{\mathcal{I}}_{r,k}(f)\right)\geq  v_k k^{-2s-2r},\quad\forall k\geq r.
\end{equation*}
Then,  
\begin{equation*}
 \frac{\widehat{\mathcal{I}}_{r,k}(f)-\mathcal{I}(f)}{\sqrt{\var\left(\widehat{\mathcal{I}}
 _{r,k}(f)
\right)}}\Rightarrow \mathcal{N}_1(0,1).
\end{equation*}
This statement also holds with $\Ihrk{r}(f)$ replaced by $\Itrk{r}(f)$.
\end{lemma}

\begin{proof}[Proof of Lemma \ref{lemma:CLT}]

We prove only the result for $\widehat{\mathcal{I}} _{r,k}(f)$, the proof for
$\widetilde{\mathcal{I}} _{r,k}(f)$ being identical. 

Let $k\geq r$ and, for all $c\in\Ck$, let
\begin{equation*}
    X_{k,c}=\frac{1}{k^s}\, h_{k, c}(U_c)
\end{equation*}
with $h_{k, c}(U_c)$ as in Lemma~\ref{lemma:main}. Note that
$\widehat{\mathcal{I}}_{r,k}(f)-\mathcal{I}(f)=\sum_{c\in\Ck}X_{k,c}$
and that $\{X_{k,c}\}_{c\in \Ck}$ is a set independent random variables for
all $k\geq r$. Then, by Lindeberg-Feller central limit theorem
(see \cite{Bill}, Theorem 27.2, page 359) to prove the lemma  it is enough to
show that, as $k\rightarrow  \infty$, 
\begin{align}\label{eq:CLT}
\frac{1}{B_k^2}\sum_{c\in \Ck}\E\left[X_{k,c}^2\ind(X_{k,c}^2>\epsilon
B_k^2)\right]\rightarrow 0,\quad\forall \epsilon>0 
\end{align}
where $B_k= \var\left(\widehat{\mathcal{I}}_{r,k}(f)
\right)^{1/2}$ for all $k$.

To show \eqref{eq:CLT} remark first that, by Theorem \ref{thm:main} and under
the assumptions of the lemma, we have
\begin{align}\label{eq:in_B}
v_k k^{-2s-2r}\leq B_k^2\leq C_{f,r,s}    k^{-s-2r}
\end{align}
where $C_{f,r,s}=\widehat{C}^2_{s,r} \|f\|^2_r$  with
$\widehat{C}_{s,r}<\infty$ as in Theorem \ref{thm:main}.

Next, let $k\geq r$ and  $c\in\Ck$, and note that, by Lemma \ref{lemma:main},
\begin{align*}
    X_{k,c}^2 = k^{-2s}\,h_{k, c}(U_c)^2 
    \leq C_{f,s,r} k^{-2s-2r},\quad \P-a.s.
\end{align*}
which, together with \eqref{eq:in_B}, implies that for all $\epsilon>0$ and $\P$-a.s.\, we have
\begin{align*}
\ind(X_{k,c}^2>\epsilon B_k^2)&\leq \ind\left(C_{f,s,r}  k^{-2s-2r}>\epsilon B_k^2
\right)\\
&\leq \ind\left(C_{f,s,r} k^{-2s-2r}>\epsilon v_k k^{-2s-2r}
\right)\\
&= \ind\left(C_{f,s,r}>\epsilon v_k
\right).
\end{align*}
Since $v_k\rightarrow\infty$,  \eqref{eq:CLT} follows and the proof is complete.
\end{proof}

By Lemma \ref{lemma:CLT}, a CLT therefore  holds for
$\widehat{\mathcal{I}}_{r,k}(f)$ and   $\widetilde{\mathcal{I}}_{r,k}(f)$ if
the variance of these estimators does not converge to zero too quickly as
$k\rightarrow\infty$. Noting that the lower bound on the variances assumed in 
Lemma \ref{lemma:CLT} converges to zero much faster than the upper bound given
in Theorem \ref{thm:main} (part \ref{res2}), we conjecture that a CLT holds in
general for $\widehat{\mathcal{I}}_{r,k}(f)$ (and
$\widetilde{\mathcal{I}}_{r,k}(f)$). 

We are able to establish this conjecture provided that the numerical
derivatives are computed in the following way. We  introduce $p_{r, k}\eqdef
\lceil k/r\rceil^s$ hyper-cubes
$\tilde{B}_q$ of volume $(r/k)^s$, $q=1,\ldots,p_{r, k}$, such that $\cup_{q=1}^{p_{r, k}}
\tilde{B}_q = [0, 1]^s$, and let $\tilde{B}_q = \cup_{j=1}^{r^s} B_k(c_j^q)$ with
$\{c_j^q\}_{q=1}^{p_{r, k}}\subset \Ck$  such that the $B_k(c_j^q)$'s are contiguous. 
Then, to each  $c\in\Ck$  we assign a $q(c)$ such that $c\in B_{q(c)}$ and impose that the numerical derivatives at $c$ are computed using
only  points $c'\in B_{q(c)}$.  The following lemma establishes that this way of
computing numerical derivatives ensures that the condition in
Lemma~\ref{lemma:CLT} is fulfilled.

\begin{lemma}\label{lemma:var}
Let $f\in\mathcal{C}^r([0,1]^s)$ for some $r\geq 2$ and, for all $k\geq r$,
$\alpha$ such that $|\alpha|<r$ and $c\in\Ck$,  let
$\widehat{D}^{\alpha}_{k}f(c)$ be as defined in Lemma
\ref{lemma:num_deriv2} with $c_j\in B_{q(c)}$ for all $j\in\{1,\dots, l_{r,\alpha}\}$. In addition, for  all $\alpha$
such that $|\alpha|=r$ let $g_\alpha:[0,1]^s\rightarrow\R$ be defined by
$g_\alpha(u)=(-1/2+u)^\alpha$, $u\in[0,1]^s$, and let
\begin{align}\label{eq:clt_condition}
\widehat{\sigma}^2_{f,r}=r^{2r+s}\sum_{|\alpha|=r}\sum_{|\alpha'|=r}\frac{\cov\left(\,\widehat{\mathcal{I}}_{r,r}(g_\alpha), \widehat{\mathcal{I}}_{r,r}(g_{\alpha'}\,)
\right)}{\alpha!\alpha'!} \int_{[0,1]^s} D^\alpha f(u)\,D^{\alpha'} f(u)\dd u.
\end{align}
Then 
\begin{equation}\label{eq:lim_var}
\lim_{k\rightarrow\infty} \left\{k^{s+2r} \var\left(\widehat{\mathcal{I}} _{r,k}(f)
\right) \right\}= \widehat{\sigma}^2_{f,r}. 
\end{equation}
The same result holds if $\Ihrk{r}(f)$ is
replaced by $\Itrk{r}(f)$.
\end{lemma}

To understand why the numerical derivatives  assumed in Lemma~\ref{lemma:var}
are convenient to show that \eqref{eq:lim_var} holds, let 
$[a,b]\subset[0,1]^s$ and  $f_{[a,b]}:[0,1]^s\rightarrow\R$ be defined by 
\begin{equation*}
    f_{[a,b]}(u)=f(a+u(b-a)),\quad u\in [0,1]^s
\end{equation*}
where, for all $u$, the product $u(b-a)$ must be understood as being
component-wise. In addition, assume that $k=mr$ for some integer $m\geq 1$, so
that the set $[0,1]^s$ can be covered by $m^s$ hypercubes
$\{\tilde{B}_q\}_{q=1}^{m^s}$ of volume $m^{-s}$. Then,   under the assumptions
on the $\widehat{D}^{\alpha}_{k}f(c)$'s made in Lemma \ref{lemma:var},
\begin{equation*}
    \widehat{\mathcal{I}} _{r,mr}(f)
    \dist \sum_{q=1}^{m^s} \widehat{\mathcal{I}} _{r,mr}(f\ind_{\tilde{B}_q})
    \dist \frac{1}{m^s}\sum_{q=1}^{m^s} \widehat{\mathcal{I}}_{r,r}(f_{\tilde{B}_q})
\end{equation*}
where the terms of the sum are independent random variables.  Since
$\widehat{\mathcal{I}}_{r,r}$ is a stochastic quadrature of degree $r-1$, it
follows from \cite[Theorem 2]{haber1969stochastic} that
\begin{equation*}
    \lim_{m\rightarrow\infty}\var\left\{ (mr)^{s/2+r} \widehat{\mathcal{I}} _{r,mr}(f)\right\}
    = \widehat{\sigma}^2_{f,r}.
\end{equation*}
Lemma \ref{lemma:var} extends this result to the case where $k$ is
not a multiple of $r$.

Combining Lemma \ref{lemma:CLT} and Lemma \ref{lemma:var} we readily obtain 
the following result.

\begin{thm}\label{thm:clt}
Let $f\in\mathcal{C}^r([0,1]^s)$ for some $r\geq 2$ and assume that, for all
$k\geq r$, $c\in\Ck$ and $\alpha$ such that $|\alpha|<r$, the numerical
derivative $\widehat{D}^{\alpha}_{k}f(c)$ and the constant
$\widehat{\sigma}^2_{f,r}$ are as defined in Lemma \ref{lemma:var}. Then, if
$\widehat{\sigma}^2_{f,r}>0$ we have
\begin{equation*}
\frac{\widehat{\mathcal{I}}_{r,k}(f)-\mathcal{I}(f)}{\sqrt{\var\left(\widehat{\mathcal{I}}
_{r,k}(f)
\right)}}\Rightarrow \mathcal{N}_1(0,1).
\end{equation*}
This statement also holds if   $\Ihrk{r}(f)$ is replaced by $\Itrk{r}(f)$.
\end{thm}

\section{Integration of vanishing functions\label{sec:vanish}}

\subsection{Principle}

We now focus on functions whose derivatives are null at the boundary of the set
$[0,1]^s$. Formally, for $r\geq 1$ we let
\begin{equation*}
    \Cfo{r}\eqdef 
    \left\{f\in \mathcal{C}^r([0,1]^s)\text{ s.t. } \max_{\alpha:|\alpha|\leq
    r} \left|D^{\alpha}f(u)\right|=0 \text{ or all } u\in\partial [0,1]^s\right\}
\end{equation*}
and consider the problem of approximating $\mathcal{I}(f)$ for
$f\in\mathcal{C}^r_0([0,1]^s)$.  Our objective is to derive an estimator that
has the same optimality properties as the estimator introduced in the previous section,
while being cheaper to compute when $f\in\Cfo{r}$. Vanishing functions may
arise for instance when performing importance sampling with a heavy-tail
proposal distribution; see the second set of numerical experiments
(Section~\ref{sub:Pima1}) for an illustration of this idea, and see Appendix \ref{sub:relevance_vanish}  for a longer discussion of the practical relevance of vanishing functions. 

We return to Haber's second estimator: 
\begin{equation*}
    \Ihrk{2}(f) = \frac{1}{k^s} \sum_{c\in\Ck} g_c(U_c), \quad U_c\sim\Uk
\end{equation*}
where, assuming $f\in \Cf{4}$, 
\begin{align*}
    g_c(u) & = \frac{f(c+u) + f(c-u)}{ 2} = f(c) + \frac 1 2 u^T H_f(c) u +
            \bigO(\|u\|^4),
\end{align*}
and $H_f(c)$ is the Hessian of $f$ at $c$. To get a smaller error, one may
combine more than two terms; e.g. with four terms: 
\begin{align*}
    \frac{ g_c(\lambda u) - \lambda^2 g_c(u)}{1 - \lambda^2} 
    & = \frac{f(c+\lambda u) + f(c-\lambda u) - \lambda^2 f(c+u) -\lambda^2
    f(c-u)}{2(1-\lambda^2)} \\
    & = f(c) + \bigO(\|u\|^4).
\end{align*}

The resulting estimator will then be a linear combination of averages of the
form
$k^{-s}\sum_cf(c+\lambda U_c)$, for a given $\lambda$. But, if $|\lambda|\neq 1$,
such an average will typically not have the desired expectation 
$\mathcal{I}(f)$, since the support of $c+\lambda U_c$ is a hyper-cube inflated 
by a factor $\lambda$. 

To address this issue, we note first that, since $f\in\Cfo{r}$, we may extend
$f$ to
$\bar{f}\in\mathcal{C}^r(\R^s)$, with $\bar{f}(u)=f(u)$ if $u\in[0, 1]^s$, and 
$\bar{f}(u)=0$ otherwise. This implies that: 
\begin{equation*}
    \mathcal{I}(f) = \int_{[0,1]^s} f(u)\dd u = \int_{\R^s} \bar{f}(u) \dd u 
    = \sum_{c\in \Cinfk} \int_{B_k(c)} \bar{f}(u) \dd u
\end{equation*}
where $\Cinfk$ is simply~\eqref{eq:def_Cmk} with $m=+\infty$; i.e. the
(infinite) set of centres of hypercubes of volume $k^{-s}$, the union of which
is $\R^s$. 

Second, if we restrict $\lambda$ to values such that $|\lambda|=1, 3, 5,
\ldots$, we observe that the support of $(c+ \lambda U_c)$ is the union of
$|\lambda|^s$ contiguous hyper-cubes in  $\Cinfk$. If we sum over $c\in
\Cinfk$, we make sure that each hyper-cube is `visited' the same number of
times. In practice, we need to consider only $c$ such that support of
$(c+\lambda U_c)$ intersects with $[0, 1]^s$, since the corresponding integral
is zero otherwise.  The following lemma 
formalises these ideas. 

\begin{lemma}\label{lemma:expectation}
Let $g\in L_1([0,1]^s)$, $\lambda\in \{\pm (2i+1),\,i\in\mathbb{N}_0\}$, $k\geq
2$, and $\bar{g}:\R^s\rightarrow\R$ be such
that $\bar{g}(u)=g(u)$ if $u\in [0,1]^s$ and $\bar{g}(u)=0$ otherwise.  Then,
\begin{equation*}
\E\left[\frac{1}{k^s}\sum_{c\in \Cmk}\bar{g}(c+\lambda U_c)\right]
=\int_{[0,1]^s}g(u)\dd u,\quad\forall  m\geq (|\lambda|-1)/2.
\end{equation*}
\end{lemma}

\subsection{Proposed estimator}

We are now able to define our vanishing estimator. Assume $r\geq 1$ is fixed,
and $f\in\Cfo{r}$.
Let $(\lambda_j)_{j=1}^\infty$ be the sequence $1, -1, 3, -3, 5, -5, \ldots$, 
and 
\begin{equation*}
m_r:=\max \{|\lambda_j|\}_{j=1}^{r}=
\begin{cases}
r,&\text{ if $r$ is odd}\\
r-1,&\text{ otherwise} 
\end{cases}
\end{equation*}
\begin{equation*}
\gamma^{(r)}\eqdef  \Gamma^{-1}_r\begin{pmatrix}
1\\
0\\
\vdots\\
0
\end{pmatrix},\quad \Gamma_r \eqdef   \begin{pmatrix}
1&1&\hdots&1\\
\lambda_1&\lambda_2&\hdots&\lambda_{r}\\
\vdots&\vdots&\vdots&\vdots\\
\lambda_1^{r-1}&\lambda_2^{r-1}&\hdots&\lambda_{r}^{r-1}
\end{pmatrix}.
\end{equation*}

The matrix $\Gamma_r$ is a Vandermonde matrix and thus, since
$\lambda_j\neq\lambda_l$ for all $j\neq l$, this matrix is invertible.
In addition, using Taylor's theorem it is easy to check that $\gamma^{(r)}$ is the vector of coefficients such that 
\begin{equation}
    \label{eq:Taylor_gcr}
    g_{r,c} (u) \eqdef \sum_{j=1}^r \gamma_j^{(r)} \bar{f}(c+\lambda_j u)
    = f(c) + \bigO(\|u\|^r). 
\end{equation}

We now define our vanishing estimator as follows: 
\begin{equation}
    \label{eq:est_vanish}
    \Ihork{r}(f) \eqdef \frac{1}{k^s} \sum_{c\in\Cmrk} g_{r,c}(U_c),\qquad
    U_c\sim\Uk.
\end{equation}

When $r=1$ or $r=2$, we recover Haber's estimators: $\Ihork{r}(f)=\Ihrk{r}(f)$
for $r=1,2$.  $\Ihork{r}(f)$ is clearly
cheaper (and simpler) to compute than the general estimator $\Ihrk{r}(f)$ of the
previous section, as the latter required computing a $\bigO(e^s)$ number of
numerical 
derivatives.  The unbiasedness of $\Ihork{r}(f)$ is a  direct consequence of
Lemma~\ref{lemma:expectation}. From~\eqref{eq:Taylor_gcr}, we see that  the
variance of $\Ihork{r}(f)$ is $\bigO(n^{-1-2r/s})$. It has therefore the  same
RMSE rate as the estimator considered in Section~\ref{sec:main}.  These
and other properties are stated in Theorem~\ref{thm:main_vanish} below. 

Before that, we must clarify what we mean by $n$ in this context. We may define
$n$ to be the number of evaluations of $\bar{f}$; in this case,
$n=r(k+2m_r)^s$, since $|\Cmrk|=(k+2m_r)^s$. Or we may define it to be the
number  of evaluations of $f(u)$ for $u\in[0, 1]^s$. In that case, $n$ is
random, with expectation $r k^s$. (To see this, apply
Lemma~\ref{lemma:expectation} to function $g(u)=1$.)  It is also bounded, i.e.
$(k-2m_r)^s \leq n \leq (k+2m_r)^s$ with probability one. Hence, whatever the
chosen definition of $n$, the statement $k=\bigO(n^{-1/s})$ remains correct.

\begin{thm}\label{thm:main_vanish}
Let $f\in\mathcal{C}_0^r([0,1]^s)$ for some $r\geq 1$. Then, for all $k\geq 2$ we have 
$\E[\widehat{\mathcal{I}}^{\,0}_{r,k}(f)]=\mathcal{I}(f)$   and there
exists a constant $\widehat{C}^{\,0}_{f,s,r}<\infty$ such that
\begin{gather*}
\E\big[ | \widehat{\mathcal{I}}^{\,0}_{r,k}(f)-\mathcal{I}(f)|^2\big]^{1/2}\leq
\widehat{C}^{\,0}_{f,s,r} n^{-\frac{1}{2}-\frac{r}{s}}, \quad \P\left(|\widehat{\mathcal{I}}^{\,0}_{r,k}(f)-\mathcal{I}(f)|\leq
    \widehat{C}^{\,0}_{f,s,r} n^{-\frac{r}{s}}
\right)=1
\end{gather*}
and such that, for all $ \delta\in(0,1)$, 
\begin{equation*}
\P\left\{ |\widehat{\mathcal{I}}^{\,0}_{r,k}(f)-\mathcal{I}(f)|\leq n^{-\frac{1}{2}-\frac{r}{s}}\, \widehat{C}^{\,0}_{f,s,r}   \sqrt{2\log\left(2 /\delta
\right)} \right\}\geq 1-\delta.
\end{equation*}
\end{thm}
\begin{proof}
As in Lemma~\ref{lemma:main}: 
for $k\geq 2$  and $c\in\mathfrak{C}_{m_r,k}$, 
let $h_{k,c}: [0,1]^s\rightarrow\R$ be defined by
\begin{equation*}
    h_{k, c}(u) \eqdef g_{r,c}(u)-\E[g_{r,c}(U_c)],\quad u\in [0,1]^s 
\end{equation*}
so that $\Ihrk{r}(f)-\mathcal{I}(f) = k^{-s}\sum_{c\in\Cmrk}h_{r,c}(U_c)$. 
(Function $h_{k, c}$ also depends on $r$ implicitly.)

Let  $u\in [-1/2k, 1/2k]^s$ and note that, from~\eqref{eq:Taylor_gcr} and the
definition of $g_{r, c}$:
\begin{equation}\label{eq:van_int}
\begin{split}
|h_{k, c}(u)|
& \leq \|f\|_r \sum_{j=1}^{r}|\gamma^{(r)}_j
\lambda_j^r|\sum_{\alpha:|\alpha|=r}
\frac{\big|u^\alpha+ \prod_{j:\alpha_j\neq 0} d_{k}(j)\big|}{\alpha !}\\
&\leq \Big(\|f\|_r \sum_{j=1}^{r}|\gamma^{(r)}_j \lambda_j^r|\Big)
\left(2^{-r}k^{-r} +k^{-r}
\right)
\end{split}
\end{equation}
where the  second inequality uses the fact that $d_{k}(j)\leq k^{-j}$ for all
$j\in\mathbb{N}$.

By \eqref{eq:van_int} there exists a constant $C<\infty$ such that,
\begin{equation*}
|h_{k,c}(u)|\leq C k^{-r},\quad\forall u\in [-1/2k, 1/2k]^s,
\quad \forall c\in\mathfrak{C}_{m_r,k},\quad\forall k\geq 2 
\end{equation*}
and thus, since by Lemma \ref{lemma:expectation} the estimator 
$\widehat{\mathcal{I}}^{\,0}_{r,k}(f)$ is unbiased, the proof of theorem 
follows from the same remarks as in the proof of Theorem~\ref{thm:main}.
\end{proof}

\section{Practical details\label{sec:pract}}

\subsection{Variance estimation\label{sub:variance}}

One advantage of the standard Monte Carlo estimator is that it is possible to
estimate its variance from a single run. It  does not seem possible to do so
with the estimators proposed in this paper. However, we highlight briefly a
method to approximate the variance from a potentially small number $l\geq 2$ of
independent runs.  This method is actually a generalisation of an approach
outlined in Section 5 of~\cite{haber1966modified} for the estimator \eqref{eq:Haber}.

Consider a generic estimator of the form:  
\begin{equation*}
 \Ih = \frac 1 n \sum_{i=1}^n Y_i
\end{equation*}
where the $Y_i$'s are independent  but not (necessarily) identically distributed.
Both estimators presented in this paper are of this form (up to some notation
adjustment); e.g. for the vanishing estimator, $Y_i$ may be identified with
$g_{r,c}(U_c)$, see~\eqref{eq:est_vanish}.

Assume we obtain $l\geq 2$ realisations of the estimator $\Ih$, based 
on independent copies $Y_n^{(j)}$ of the $Y_n$.
Since  
\begin{equation*}
\var(\Ih) = \frac{1}{n^2} \sum_{i=1}^n \var(Y_i)    
\end{equation*} 
take, as an estimator of $\var(\Ih)$, 
\begin{equation*}
    \hat V \eqdef \frac{1}{n^2 } \sum_{i=1}^n \frac{1}{l-1} \sum_{j=1}^l
(Y_i^{(j)}-\bar{Y}_i)^2,\qquad \bar{Y}_i \eqdef \frac 1 l \sum_{j=1}^l Y_i^{(j)}.
\end{equation*}

It is easy to establish that, for  the two types of estimators introduced 
in this paper, $\Ihrk{r}(f)$ and $\Ihork{r}(f)$ (for a given $r\geq 1$), one has, for a
fixed $l\geq 1$:
\begin{equation*}
    \var(\hat V) = \bigO(n^{-3-4r/s})
\end{equation*}
which is $n^{-1}$ smaller than the square of $\bigO(n^{-1-2r/s})$, the rate at
which the true variance goes to zero.

In other words, estimator $\hat V$ will have a small relative
error as soon as $n$ is large (even for a small $l$). Of course, if we generate
$l$ independent realisations of a given estimator (preferably in parallel),
then we should return as a final estimate the average of these $l$
realisations, together with an estimate of its variance, that is, $\hat{V} /
l$. 

\subsection{Automatic order selection for the vanishing
estimator\label{sub:compute_all_orders}}

Given~\eqref{eq:Taylor_gcr} and \eqref{eq:est_vanish}, we may rewrite the
vanishing estimator as follows:
\begin{align*}
    \Ihork{r}(f) 
      = \sum_{j=1}^r \gamma_j^{(r)} 
    \left\{ \frac{1}{k^s} \sum_{c\in\Cmrk} \bar{f}(c+\lambda_j U_c) \right\}  
     = \sum_{j=1}^r \gamma_j^{(r)} 
    \left\{ \frac{1}{k^s} \sum_{c\in\Cx{m_j, k}} \bar{f}(c+\lambda_j U_c) \right\}
\end{align*} 
where in the second line we use the fact that $\bar{f}(c+\lambda_j U_c)=0$
whenever $c\notin \Cx{m_j, k}$. 

We may pre-compute the $r$ averages above, and use them to compute
simultaneously $\Ihork{r'}(f)$ for $r'=1,\ldots,r$, at (essentially)  the same
cost as computing only $\Ihork{r}(f)$. If we generate several copies of
these estimators, we may then choose the value $r'$ with the smallest estimated
variance (using the variance estimator proposed in the previous section). 
We may use a similar approach for the non-vanishing estimator $\Ihrk{r}(f)$, but
in that case there does not seem to be any short-cut for computing
simultaneously $\Ihrk{r}(f)$ for different values of $r$.

\section{Numerical experiments\label{sec:numerics}}

In this section, we assess and compare estimators of expectations
$\mathcal{I}(f)$ as follows. For a fixed function $f:[0,1]^s \rightarrow \R$
and a range of values for $k$, we generate 50 independent copies of the
considered estimators, and produce plots where:
\begin{itemize}
    \item the $x-$axis is the number of evaluations of  $f$. When this
        quantity is random (vanishing estimator), we report the average over
        the independent runs. 

  \item the $y-$axis is a measure of the relative error; that is, either the 
      mean squared error (MSE) divided by the true value of $\mathcal{I}(f)$, when
      this quantity is known, or the empirical variance divided by the square 
      of the average, when it is not.  In the former (resp. latter) case, the
      label of the $y-$axis is \textsf{rel-mse} (resp. \textsf{rel-var}). In
      both cases, we discard results where the relative error is too close to
      machine epsilon (i.e. when MSE or variance is not $\gg 10^{-32}$ ). In
      such cases, the corresponding estimates may be considered as exact (up to
      machine epsilon). 
\end{itemize}

It is customary in this type of plot to overlay a straight line that
corresponds to the expected rate, i.e. $\bigO(n^{-1+2r/d})$ for our estimators.
(The log-scale is used on both axes.) However, in our case, the performance of
our estimators (often) matches closely these rates, making these lines
hard to distinguish. For this reason we do not plot them in what follows.

An open-source python package implementing the two proposed estimators and the
following  numerical experiments may be found at
\url{https://github.com/nchopin/cubic_strat}. 
The numerical derivatives that appear in the control variates of the
non-vanishing estimator were computed by the 
the findiff package of \cite{findiff}. We note that the numerical derivatives computed with this package are not implemented   in a way which ensures that we have a CLT for $\Ihrk{r}(f)$ (that is, they do not verify the assumptions of Theorem~\ref{thm:clt}).

\subsection{Comparison between the non-vanishing estimator and Dick's
estimator}

As mentioned in the introduction,  in \cite{dick2011higher} Dick introduced 
higher-order estimators of $\mathcal{I}(f)$ (henceforth, Dick's estimators),
based on scrambled digital nets,  which achieve
$\bigO(n^{-1/2-\alpha+\epsilon})$ RMSE for functions
$f\in\mathcal{D}^\alpha([0,1]^s)$, $\alpha\geq 2$, the set of functions such
that all partial derivatives obtained by differentiating with respect to each
variable up to $\alpha$-times is square integrable. When $s\geq 2$, this 
estimator does not require the existence of the same number of partial
derivatives as our stratified estimators (even if we set $r=s \times \alpha $).
For instance, for $s=2$, denoting $u=(x, y)$, Dick's estimator requires the
existence of $\partial f/\partial x$, $\partial f/\partial y$ and $\partial^2
f/\partial x\partial y$ at order $\alpha=1$, while our stratified estimator
requires only the first two when $r=1$; or, alternatively,  these three
derivatives plus  $\partial^2 f /\partial x^2$, $\partial^2 f/ \partial y^2$ at
order $r=2$. This technical point should be kept in mind in the following
comparison, where Dick's estimator is implemented using the  Sobol' sequence as underlying digital sequence. 

We consider the following functions: for $s=1$, $f_1(u)=ue^u$, and for $s\geq
2$, 
\[ f(u) = \left(\prod_{j=1}^s u_j^{j-1} \right) \exp\left( \prod_{j=1}^s u_j \right).
\]
Note that $\mathcal{I}(f_1)=1$, and $\mathcal{I}(f_s)=e-\sum_{j=0}^{s-1}
(1/j!)$ for $s\geq 2$. The aforementioned paper used the first two functions of
this sequence to illustrate the numerical performance of Dick's estimators.   We
compare the performance of Dick's higher-order estimators (for $\alpha=1, 2, 3,
4$) with our non-vanishing estimator (for $r=1, 2, 4, 6, 8$, and, in addition,
$r=10$ for $s=1$ and $s=2$); see Figures \ref{fig:dick12} and \ref{fig:dick46}. 

\begin{figure}
    \centering
    \includegraphics[width=0.40\textwidth]{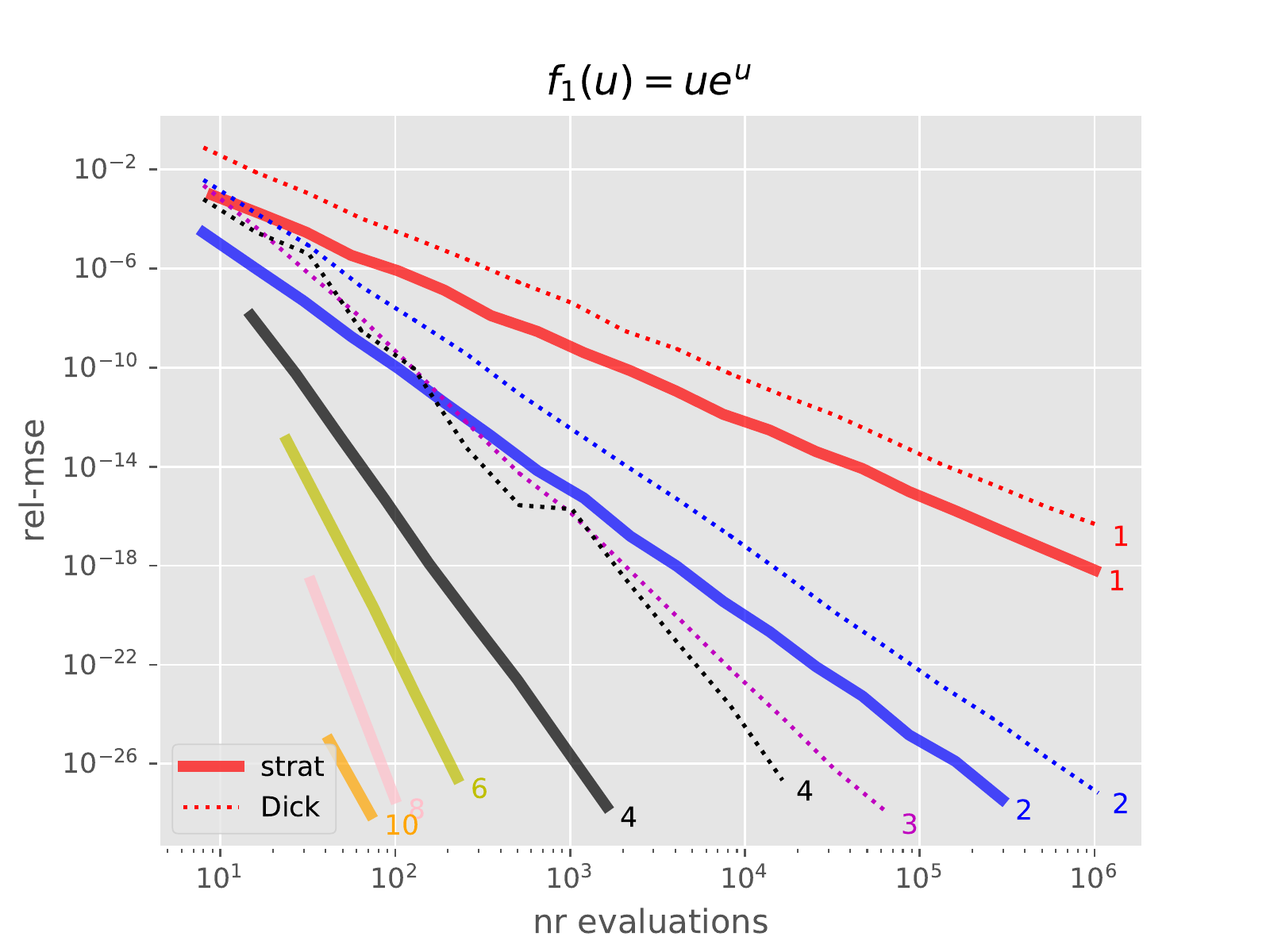}
    \includegraphics[width=0.40\textwidth]{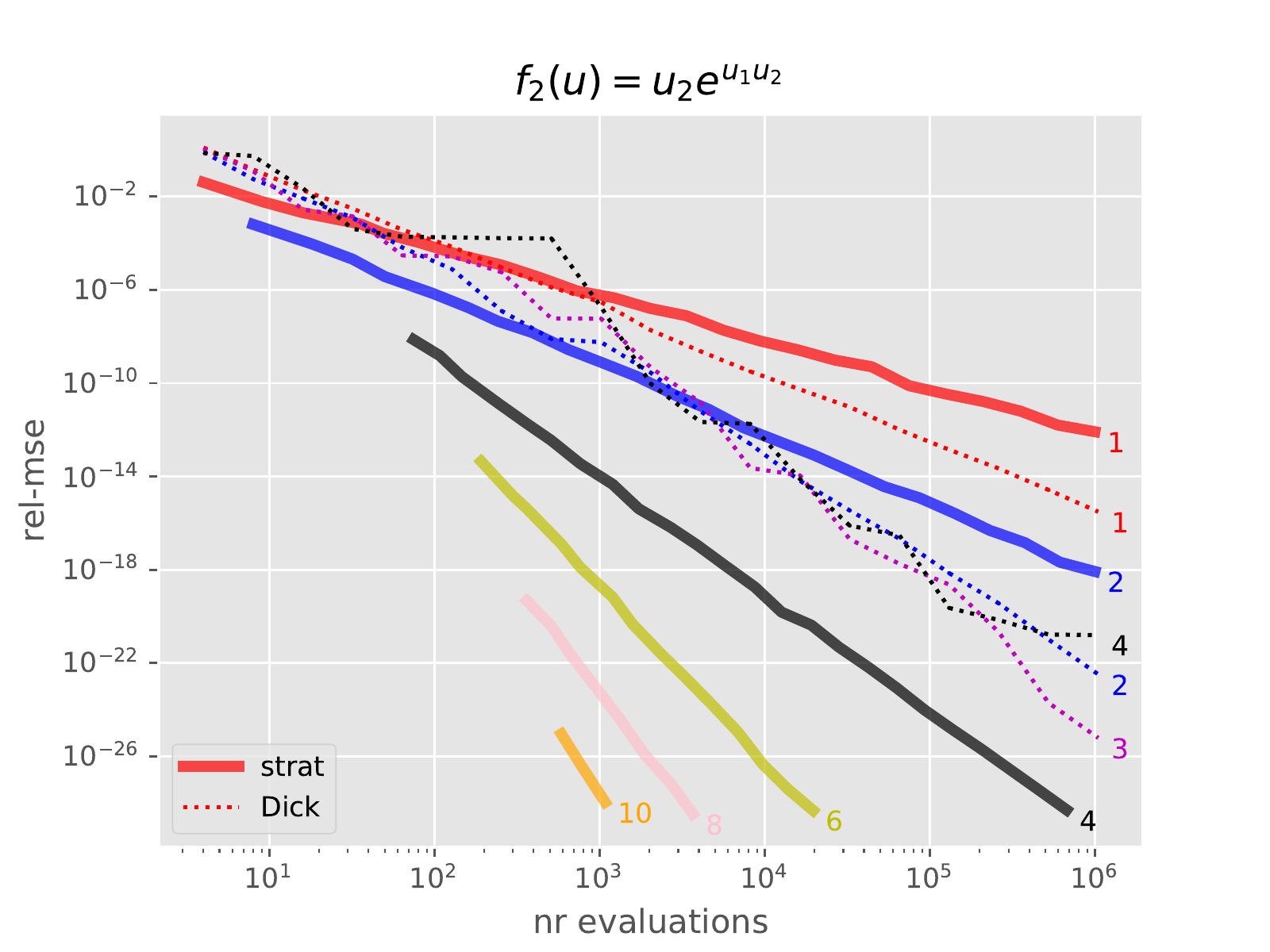}
    \caption{Relative MSE (mean squared error) vs number of evaluations for the
    vanishing estimator (thick lines) and Dick's estimator (dotted line). The
    value of $r$ (stratified) or $\alpha$ (Dick's) are printed next to each
    curve.  Left: $f_1$; Right: $f_2$. } 
    \label{fig:dick12}
\end{figure}

For $s=1$ (left panel of Figure \ref{fig:dick12}), both estimators require
exactly the same number of derivatives, hence the comparison is
straightforward. Both estimators show the expected MSE rate,
$\bigO(n^{-1+2r})$, (taking $\alpha=r$); on the other hand, the stratified
estimator seems to consistently have lower MSE.

\begin{figure}
    \centering
    \includegraphics[width=0.40\textwidth]{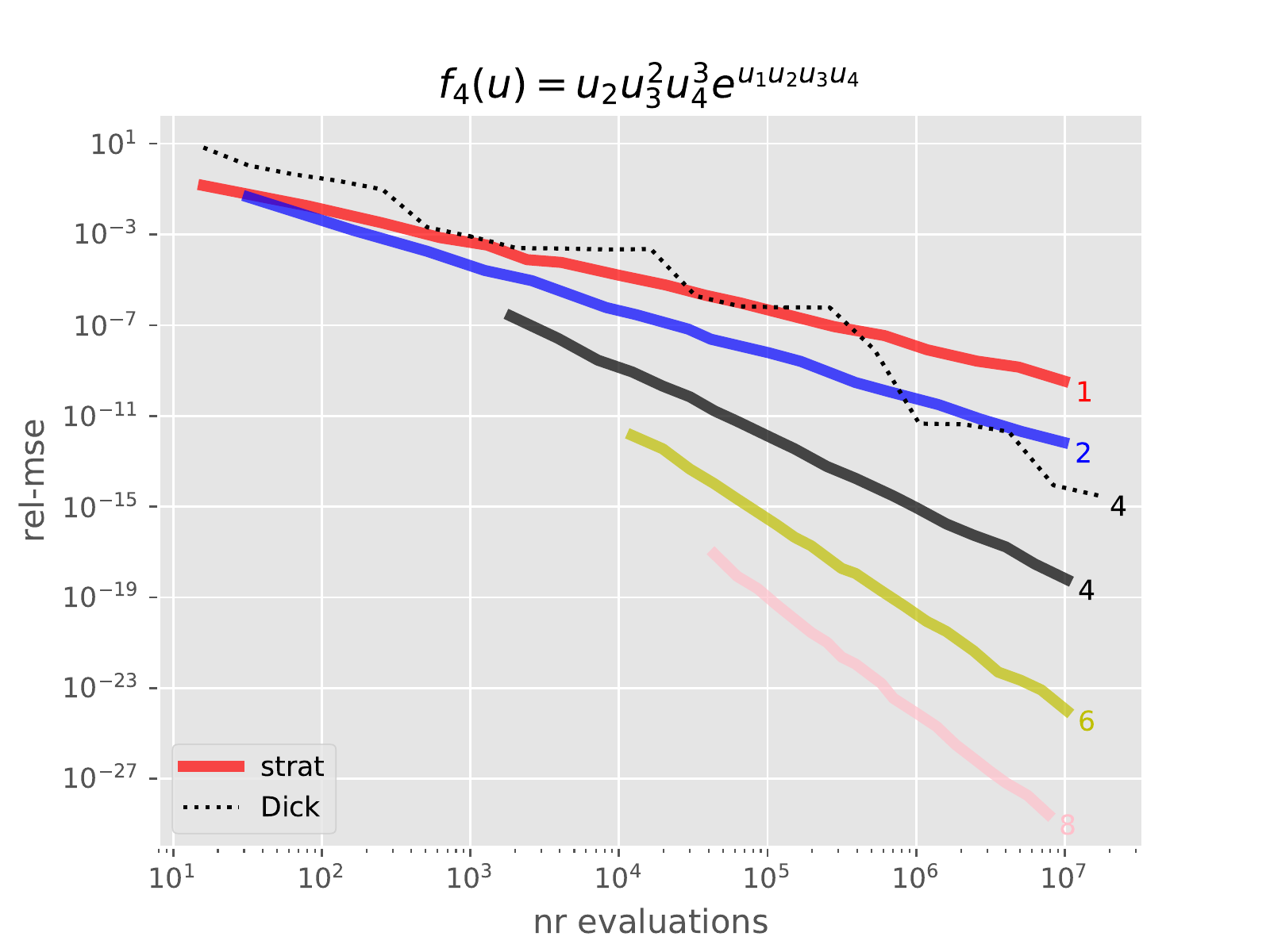}
    \includegraphics[width=0.40\textwidth]{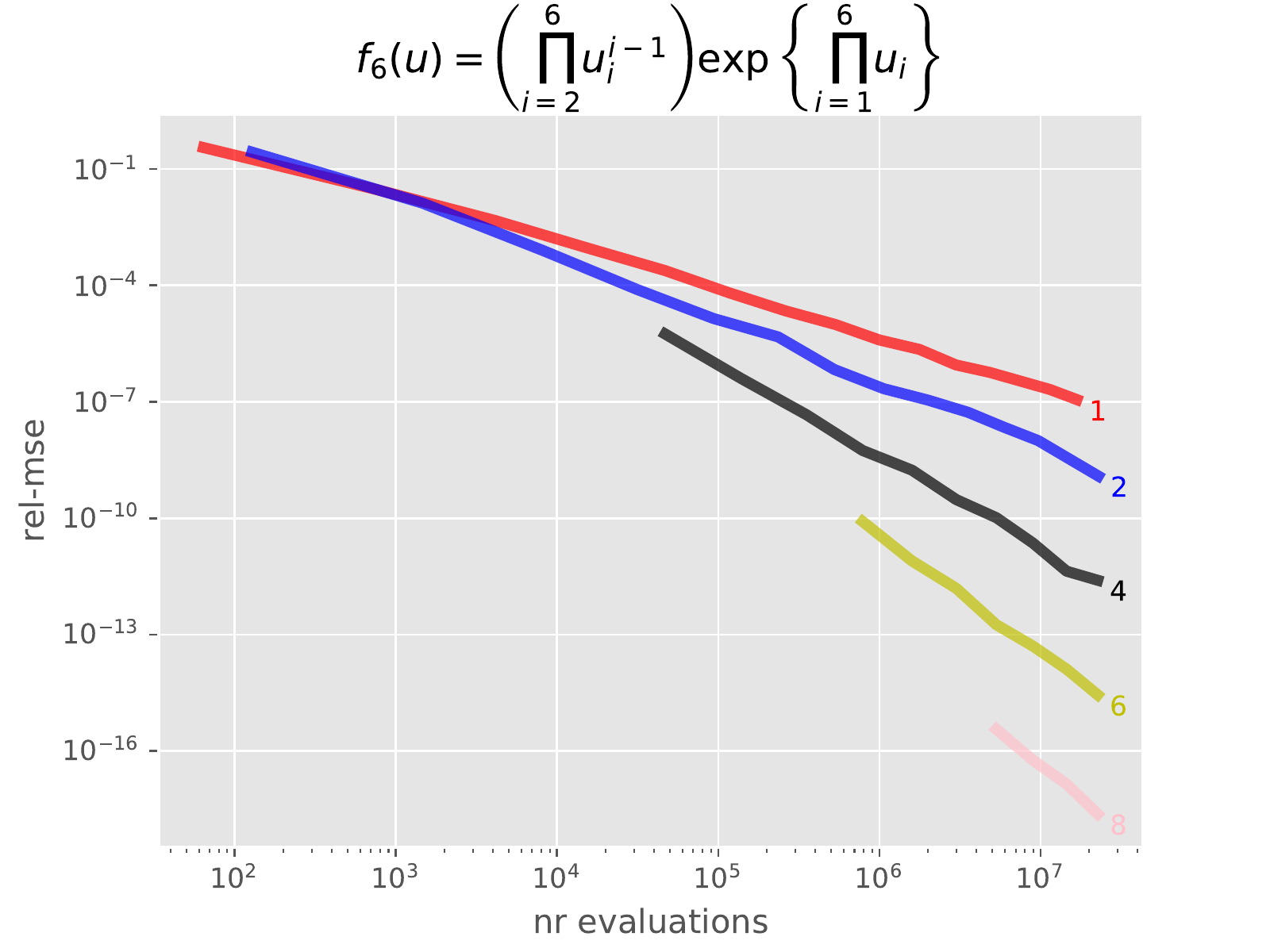}
    \caption{Same plot as in Figure \ref{fig:dick12} for functions $f_4$ (left)
    and $f_6$ (right).}
    \label{fig:dick46}
\end{figure}

For $s=2$ (right panel of Figure \ref{fig:dick12}), the comparison becomes less
straightforward, as we explained above. The fact that Dick's estimator shows
intermediate performance between the stratified estimators for $r=1$ and $r=2$
is reasonable, since it requires strictly more partial derivatives than for
$r=1$, and strictly less than for $r=2$; as discussed in the example above. On
the other hand, Dick's estimator at order $\alpha=4$ seems outperformed by 
both the same estimator at orders $\alpha=2$ and 3, and the stratified
estimator at order $r=4$. This is despite the fact that Dick's estimator  with
$\alpha=4$ requires strictly more partial derivatives than the stratified
estimator with $r=4$. This suggests that, when $\alpha$ increases, Dick's
estimator requires a larger and larger number of evaluations before exhibiting
the expected rate of convergence. 

For $s=4$ (left panel of Figure \ref{fig:dick46}), we plot only the relative
MSE of Dick's estimator for $\alpha=4$. Again, we observe the same phenomenon:
i.e. even with $10^7$ evaluations it is not yet competitive with the stratified
estimator (with $r=4$) despite requiring more partial derivatives. 

In all these plots, the MSE of the proposed estimator matches very closely the
expected rate. On the other hand, recall that, for $r\geq 4$, the estimator
requires $3 k^s$ evaluations of $f$, and is properly defined only for $k\geq
r$. (In addition, the way the numerical derivatives are computed in package
findiff imposes that $k\geq 3r/2 - 1$.) This implies that this estimator is
only defined for a large number of evaluations when $r$ and $s$ are large, as
shown in Figures \ref{fig:dick12} and \ref{fig:dick46}. This is of course a
limitation of the non-vanishing estimator. We shall see that the vanishing
estimator is less affected by this issue; i.e. it may be computed for smaller
values of $n$.

\subsection{Vanishing estimator: Bayesian model choice\label{sub:Pima1}}

We now consider a class of vanishing functions in order to assess our vanishing
estimator. We construct these functions so that their integral equals the
marginal  likelihood $\int p(\beta) L(y|\beta)\dd\beta$ of a Bayesian statistical
model, where $\beta\in \R^s$, $p(\beta)$ is a Gaussian prior density (with mean
0, and covariance $5^2 I_s$), $L(y|\beta)$ is the likelihood of a logistic
regression model: $L(y|\beta) = \prod_{i=1}^n F(y_i \beta^T x_i)$,
$F(z)=1/(1+e^{-z})$, and the data $(x_i, y_i)_{i=1}^{n}$ consist of predictors
$x_i\in\R^s$ and labels $y_i\in\{-1, 1\}$. 

We adapt the importance sampling approach described in \cite{mr3634307} to 
approximate such quantities as follows: we obtain numerically the mode
$\hat\beta$, and the Hessian at $\beta=\hat\beta$, of the function $h(\beta)=\log\{
p(\beta)L(y|\beta)\}$; hence 
$h(\beta)\approx h(\hat\beta)- (1/2) (\beta-\hat\beta)^T H(\beta
-\hat\beta)$. Then we set $f(u) = \exp\{h(\hat\beta+L\psi_s(u))\}$, with $L$
the Cholesky lower triangle of $H$, $LL^\top = H$, and $\psi_s$  the function defined in Appendix \ref{sub:relevance_vanish} (for $\tau=1.5$), which maps   $(0,1)^s$ into $\R^s$.

As in \cite{mr3634307}, we consider the Pima dataset (which has 10 predictors,
if we include an intercept). More precisely, for $s=2$, 4, 6, and 8, we take
the first $s$ predictors, and compute the corresponding marginal likelihoods. 
Note that computing these quantities for all possible subsets of the predictors
is a standard way to perform variable selection in Bayesian inference. 

Figure \ref{fig:pima} showcases the performance of the vanishing estimators for
$s=2$ to 8 and at orders 1 to 10 (for $s=2$ and $s=4$), 8 (for $s=6$), and 4
(for $s=8$). Results for higher orders are not displayed for $s=6$ and $s=8$
because they did not lead to lower variance even for the highest values of
number of evaluations. 

Note the slightly different behaviour relative to the previous example. The
vanishing estimator is defined for lower  numbers of evaluations. On the other
hand, it exhibits the expected rate only for a large enough number of
evaluations. As expected, the relative gain obtained by increasing $r$
decreases with the dimension (and requires a larger and larger number of
evaluations to appear clearly). 

Notice that, in Figure \ref{fig:pima}, the number of evaluations
has a different range for different values of $r$. This is because the number
of evaluations at order $r$ is $rk^s$, and we considered the same range of
values for $k$. It was convenient to do so, because, as explained in
Section~\ref{sub:compute_all_orders}, it is possible to compute simultaneously
the vanishing estimators at orders 1 to, say $r_{\max}$ (using the same random
numbers), at the cost of obtaining the estimator at highest order, $r_{\max}$. 

\begin{figure}[htpb]
    \centering
    \includegraphics[width=0.49\textwidth]{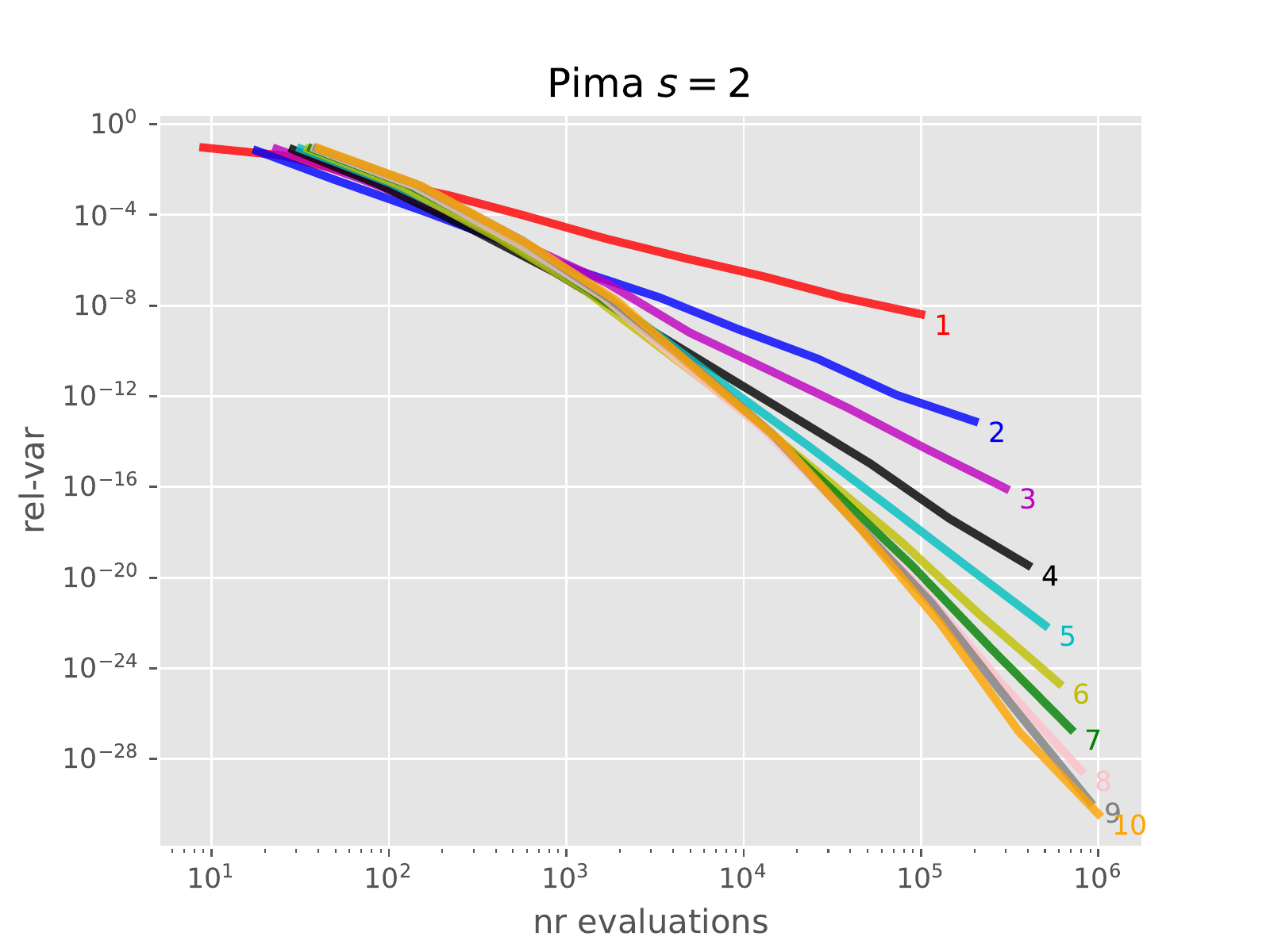}
    \includegraphics[width=0.49\textwidth]{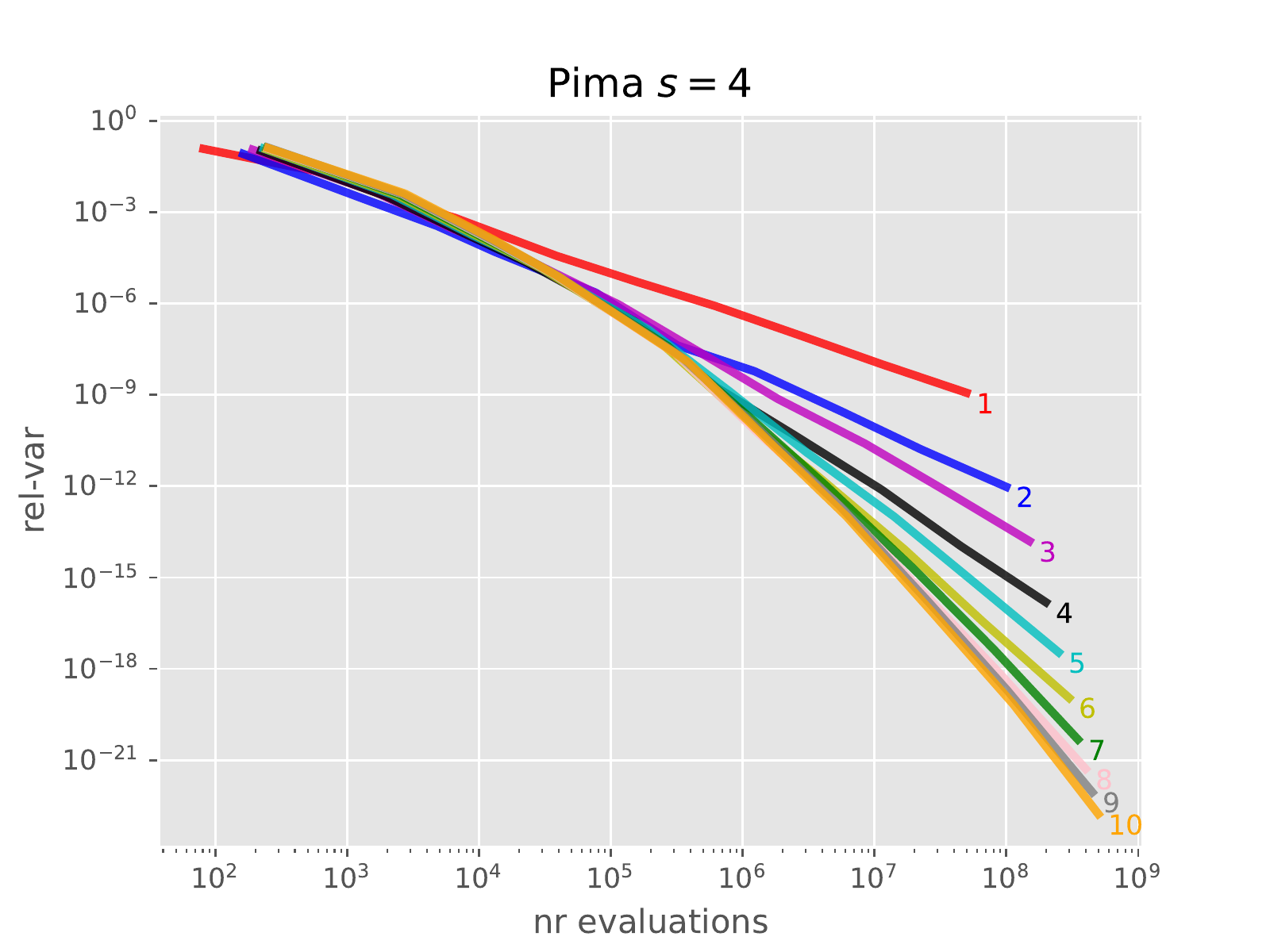}
    \includegraphics[width=0.49\textwidth]{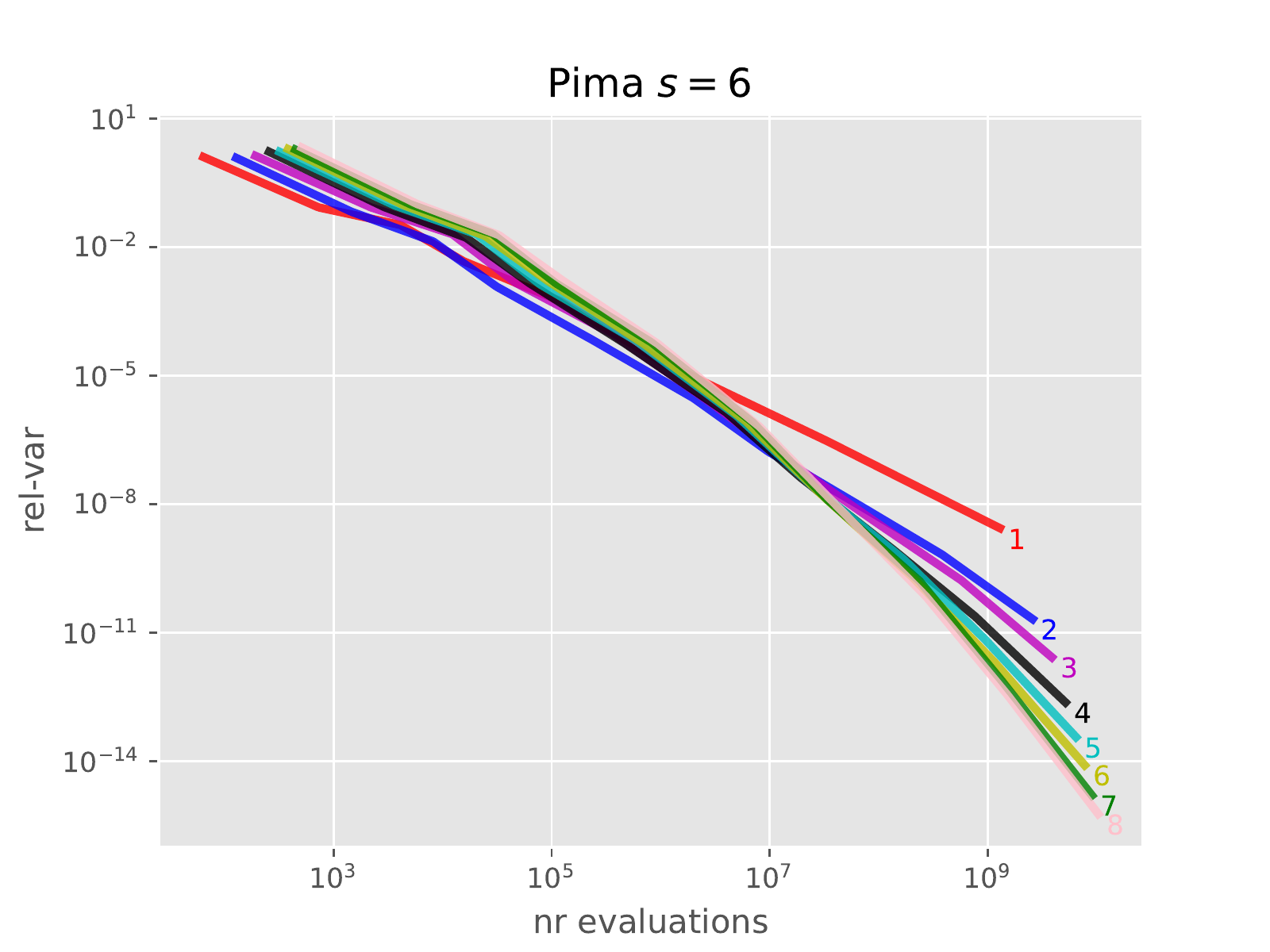}
    \includegraphics[width=0.49\textwidth]{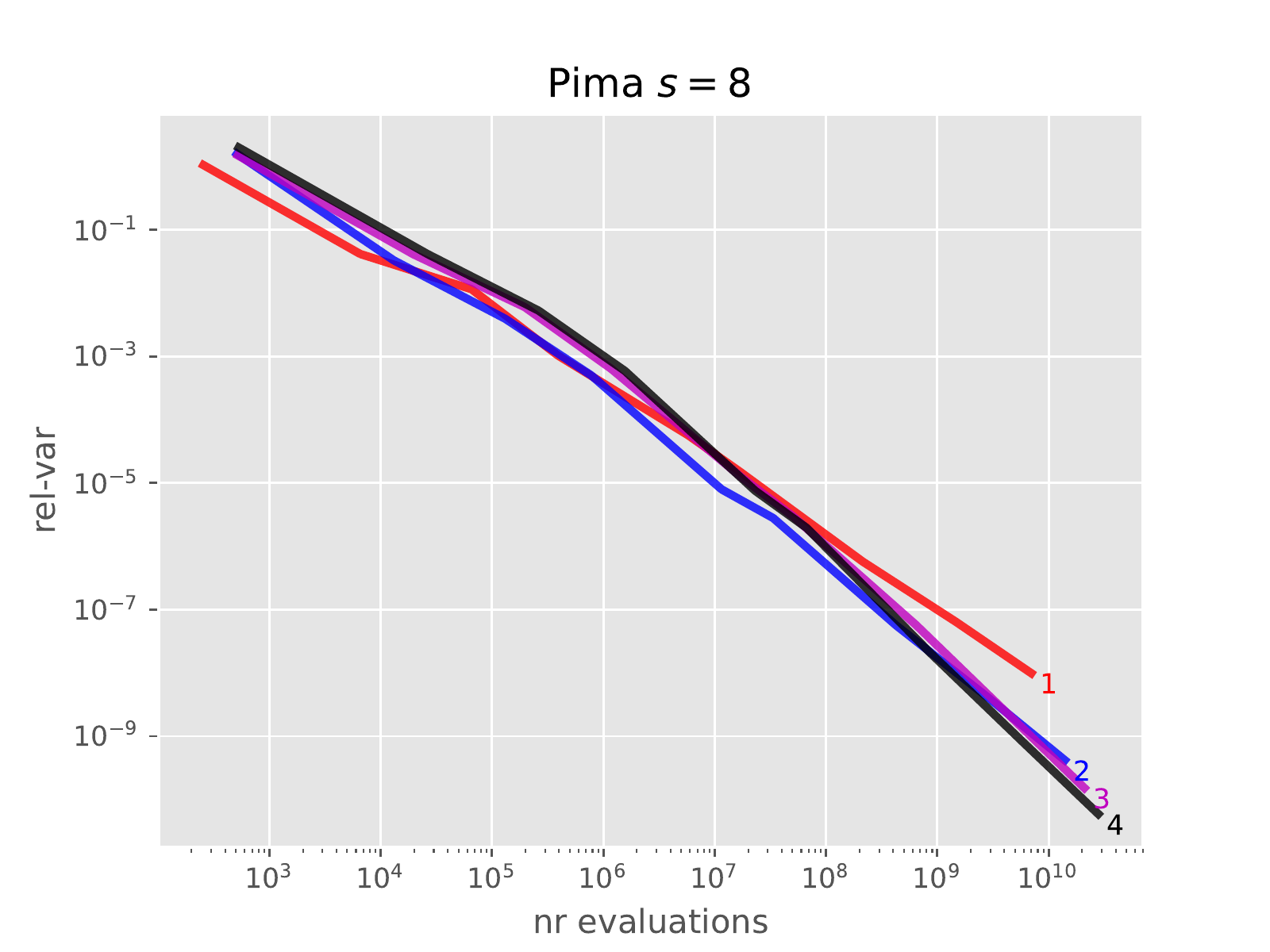}
    \caption{Relative variance of the vanishing estimator versus number of
    evaluations for Pima example, with $s=2$, 4, 6, 8.}
    \label{fig:pima}
\end{figure}

See Appendix \ref{sub:Pima2} for a comparison of the non-vanishing and vanishing
estimators on this example.

\section{Future work}\label{sec:conclusion}

The main limitation of cubic stratification is that it cannot realistically work for
$s  \gg 10$, since the number of cubes required to partition $[0,1]^s$ is
$k^s$. We could use rectangles instead, and take
$n=\prod_{i=1}^s k_i$, with $k_i$ smaller (or even $=1$) when $f$ is nearly
constant in component $i$, a bit in the spirit of \cite{sloan1998quasi}. 
Determining how we could choose the $k_i$ in a meaningful way is left for
future work.

\section*{Acknowledgments}
The authors wish to thank 
Adrien Corenflos,
Erich Novak, and 
Art Owen
for helpful remarks on a preliminary version on this manuscript. 

\bibliographystyle{siamplain}
\bibliography{complete}

\appendix

\section{Relevance of vanishing functions\label{sub:relevance_vanish}}

Consider the problem of approximating the integral of a function $g$ over
$\R^s$. A common strategy is to rewrite this integral as an expectation
with respect to a chosen, $[0, 1]^s$-supported distribution; and 
then use Monte Carlo to approximate it. 
Since $\lim_{\|x\|\rightarrow\infty} g(x)=0$, this expectation will often be an 
integral of a vanishing function. Thus, one may use instead our vanishing
estimator to approximate the integral of interest.  

The following lemma outlines a particular recipe to rewrite an integral over
$\R^s$ into the integral of a vanishing function. We designed this recipe to
make sure that the conditions on $g$ (to ensure that the transformed integrand
is indeed vanishing) are weak; essentially $g$ and its derivatives must decay
at polynomial rates at infinity. The rewritten integral is an expectation with
respect to a `Student-like' distribution, with heavy tails, whose Rosenblatt
transformation is given by $\psi$ below.

\begin{prop}\label{prop:Rs}

Let $r\geq 1$, $g\in \mathcal{C}^r(\R^s)\cap L_1(\R^s)$ be such that
\begin{align}\label{eq:cond_g}
    \lim_{\|x\|\rightarrow \infty}\left( \max_{\alpha: |\alpha|\leq r}
        D^\alpha g(x) \prod_{i=1}^s  |x_i|^{c}\right)=0, \quad\forall c>0
\end{align}
and, for some $\tau>0$,  let $\psi_s:\R^s\rightarrow(0,1)^s$ be the
$\mathcal{C}^r$-diffeomorphism defined by
\begin{equation*}
    \psi_s(u)=\left(\frac{2u_1-1}{u_1^\tau(1-u_1)^\tau},\dots,\frac{2u_s-1}{u_s^\tau(1-u_s)^\tau}\right),\quad   u\in (0,1)^s,
\end{equation*}
and let $f_{g,\psi}:[0,1]^s\rightarrow \R$  be defined by
\begin{equation}\label{eq:f_trans}
    f_{g,\psi}(u)=
        g\left(\psi_s(u)
\right)\prod_{i=1}^s\left(\frac{2}{u_i^\tau(1-u_i)^\tau}+\frac{\tau
    (2u_i-1)^2}{u_i^{\tau+1}(1-u)^{\tau+1}}\right).
\end{equation}
Then, $f_{g,\psi}\in\mathcal{C}_0^r([0,1]^s)$ and
$\mathcal{I}(f_{g,\psi})=\int_{\R^s}g(x)\dd x$.
\end{prop}

\begin{proof}
We have
\begin{align}\label{eq:comp_deriv0}
D^\alpha f_{g,\psi}(u) =\sum_{\nu\in \mathcal{N}_\alpha} D^{|\nu|} (g\circ
\psi_s)(u)\prod_{i=1}^s \frac{\dd^{\alpha_i-\nu_i}}{\dd u_i^{\alpha_i-\nu_i}}
\psi_1(u_i),\quad\forall u\in (0,1)^s
\end{align}
where 
\begin{equation*}
\mathcal{N}_\alpha=\{\nu\in\mathbb{N}_0^s:\,\,\nu_i\in\{0,\alpha_i\},\,i=1,\dots,s\}.
\end{equation*}

By  \cite[Theorem 1]{constantine1996multivariate}  for all $\nu\in\mathbb{N}^s$ we have
\begin{equation}\label{eq:comp_deriv}
\begin{split}
&\frac{D^{\nu} (g\circ\psi_s)(u)}{\nu!}\\
&=\sum_{\lambda\in\mathbb{N}_0^s:\,  |\lambda|\leq |\nu|} \left(D^\lambda
g
\right)(\psi_s(u))  \sum_{l=1}^{|\lambda|}\sum_{(\gamma,\beta) \in
p_l(\nu,\lambda)}\prod_{j=1}^l\frac{1}{(\beta!)(\gamma!)^{|\beta|}}\prod_{i=1}^s
\left(\frac{\dd^{\gamma_{ij}}}{\dd u_i^{\gamma_{ij}}}\psi_1(u_i)\right)^{\beta_{ij}}
\end{split}
\end{equation}
where, for all $\lambda\in\mathbb{N}_0$ with $|\lambda|\leq |\nu|$, the set $p_l(\nu,\lambda)\subset\mathbb{N}_0^s\times\mathbb{N}_0^s$ is as defined in \cite[Theorem 1]{constantine1996multivariate}.

On the other hand, it is easily checked that, as $u\rightarrow u'\in\{0,1\}$,  
\begin{align*}
\frac{\dd^a \psi_1^{-1}(u)}{\dd u^a }=\mathcal{O}\left(\left(u(1-u)
\right)^{-(a+\tau)}\right),\quad\forall a\in\mathbb{N}_0 
\end{align*}
which, together with \eqref{eq:comp_deriv0}-\eqref{eq:comp_deriv}, shows the result.
\end{proof}
\begin{remark}
Condition \eqref{eq:comp_deriv}  on $g$ is stronger than needed. Indeed, given a value of $\tau>0$, for the conclusion of Proposition \ref{prop:Rs} to hold it is enough that 
\begin{equation*}
\lim_{\|x\|\rightarrow \infty}  \max_{\alpha: |\alpha|\leq r}    D^\alpha g(x) \prod_{i=1}^s  |x_i|^{c_{r,s,\tau}}=0
\end{equation*}
for  some constant $c_{r,s,\tau}<\infty$. From the proof of the proposition we note that $c_{r,s,\tau}$ decreases with $\tau$.
\end{remark}


See also our second set of numerical experiments (Section~\ref{sub:Pima1})
for an application of this recipe to the computation of the marginal likelihood
in Bayesian inference.

\section{Comparing the non-vanishing and the vanishing  estimators\label{sub:Pima2}}

\begin{figure}[htpb]
    \centering
    \includegraphics[width=0.49\textwidth]{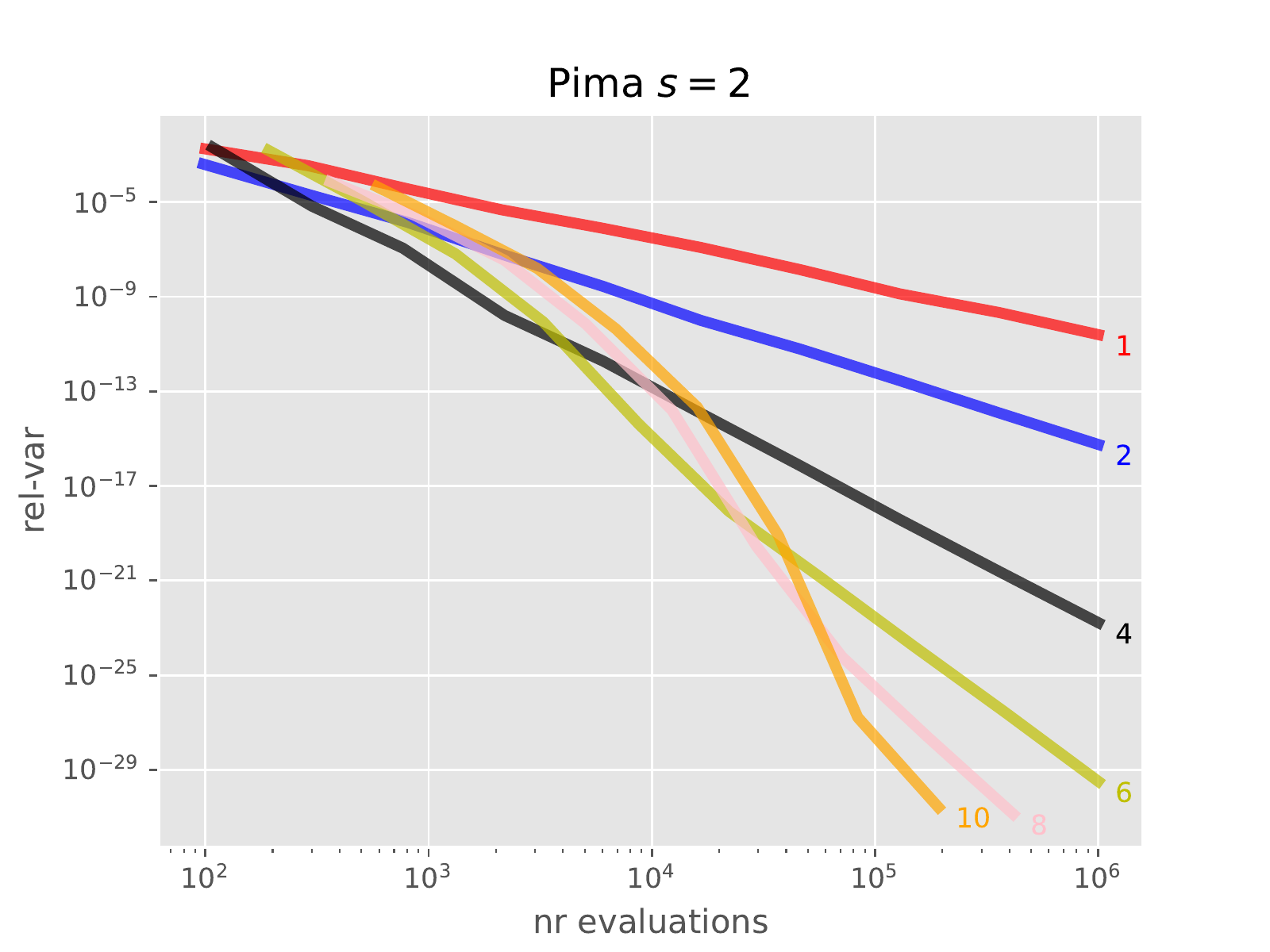}
    \includegraphics[width=0.49\textwidth]{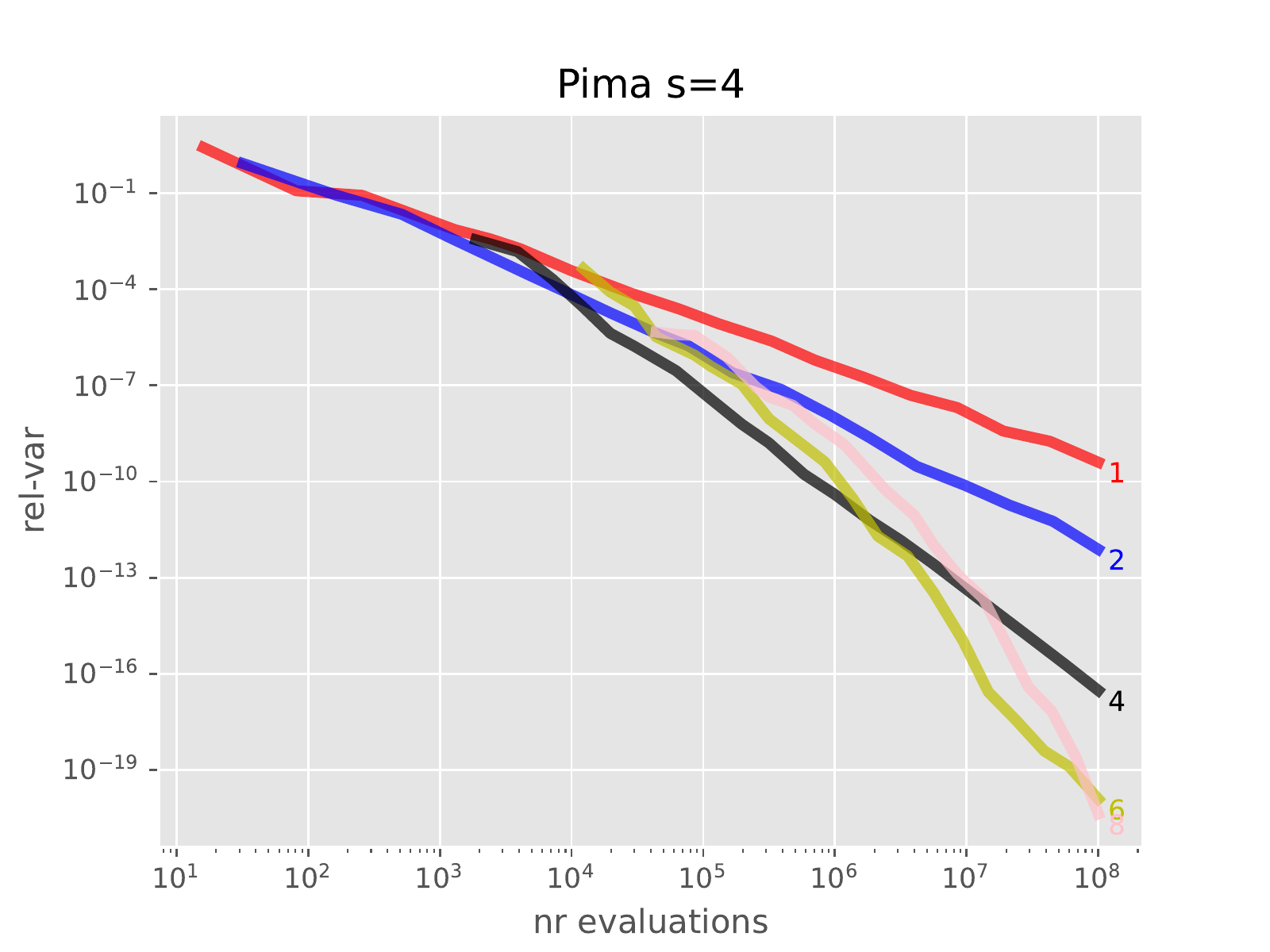}
    \caption{Relative variance of the non-vanishing estimator versus number of
    evaluations for Pima example when $s=2$ (left) and $s=4$ (right).} 
    \label{fig:nvpima}
\end{figure}

When a function $f$ is vanishing, one may use either a vanishing estimator
$\Ihork{r}(f)$ or a non-vanishing estimator $\Ihrk{r}(f)$ to compute its integral.
One may wonder which type of estimators may lead to better performance.
Figure~\ref{fig:nvpima} showcases the performance of the non-vanishing
estimator when applied to the functions of the previous example for $s=2$ and
$s=4$, and should be compared to the top panels of  Figure~\ref{fig:pima}.

One sees that, in this particular case, we do obtain better performance with 
the non-vanishing estimator  for $s=2$. (The picture is less clear for $s=4$.)
On the other hand, note that the non-vanishing estimator is less convenient to
use. As we explained in the previous example and in
Section~\ref{sub:compute_all_orders}, one can compute simultaneously the
vanishing estimators at orders 1 to some $r_{\max}$. It is then possible to
select the order that leads to best performance (using the variance estimator
described in Section~\ref{sub:variance}). On the other hand, the left panel of
Figure~\ref{fig:nvpima} shows clearly that one does not know in advance which
value of $r$ may lead to best performance when using the non-vanishing
estimator.

\section{Proofs}\label{app-proofs}

\subsection{Proof of Lemma~\ref{lemma:num_deriv2}\label{app:p-lemma:num_deriv2}}

We consider first the univariate case: $s=1$, $g\in\mathcal{C}^r([0, 1])$. 
Let
\begin{equation*}
    \Ck^{(1)} \eqdef \left\{  \frac{2j+1}{2k} \text{ s.t. }\,j\in\{0,\dots,k-1\}\right\}
\end{equation*}
which is $\Ck$ when $s=1$, and 
let $l\geq 2$ be an integer, $k\geq l$, 
\begin{equation*}
S_l \eqdef \left\{\kappa\in\{-l+1,\dots,l-1\}^l:\,\kappa_i\neq \kappa_j,\,\,\forall
i\neq j\right\}
\end{equation*}
and
$
V_l \eqdef \left\{A^{-1}_\kappa e^{(a)}:\,\ \kappa\in
S_l,\,a\in\{1,\dots,l-1\}\right\}
$
with $A_\kappa$ and $e^{(a)}$  as defined in Lemma \ref{lemma:num_deriv} (with
$r=l$). 

Let 
$$
 \tilde{C}_l=\max_{\{w_j\}_{j=1}^l\in V_l,\,\{\kappa_j\}_{j=1}^k\in S_l}\sum_{j=1}^l |w_j\kappa_j^l|.
 $$
 Then, by Lemma
\ref{lemma:num_deriv},    for all  $c'\in\Ck^{(1)}$,  all $\kappa\in
S_l$ such that $c'+\kappa_j/k\in [0,1]$ for all $j\in\{1,\dots,l\}$, and  all
$a\in\{1,\dots,l-1\}$,  there exists a set $\{w^{(a,c')}_j\}_{j=1}^l\in V_l$
such that
\begin{equation} 
    \label{eq:toShowCor}
\left|g^{(a)}(c')-\frac{\sum_{j=1}^l w^{(a,c')}_j
g(c'+\kappa_j/k)}{k^{-a}}\right|
\leq k^{-(l-a)}\|g\|_l \tilde{C}_l.
\end{equation}

We let $\mathcal{W}_r=\cup_{j=2}^{r} V_j$. We now consider the
multivariate case, $s\geq 2$, and prove the lemma  by induction on $|\alpha|_0$. 

To this aim, let $\alpha$ be such that $|\alpha|_0=1$, $c=(c_1,\dots,c_s)\in
\Ck$, $p\in\{1,\dots,s\}$ such that $\alpha_{p}=1$ and $g_c\in
\mathcal{C}^r([0,1])$ defined as (with obvious convention when
$p\in\{1,s\}$)
\begin{equation*}
g_c(c') \eqdef f\left(c_1,\dots,c_{p-1}, c', c_{p+1},\dots, c_s\right),
\quad \forall c'\in[0,1].
\end{equation*}

Next, let $\{c'_j\}_{j=1}^r$ be $r$ distinct elements of $\Ck^{(1)}$ such that
$|c_{p}-c'_j|\leq (r-1)/k$ for all $j\in\{1,\dots,r\}$, and let
$\kappa^{\alpha}_j=k(c_j-c_{p})$ for all $j$. Note that the resulting
vector $\kappa^{\alpha}$ is such that $\kappa^{\alpha}\in S_r$. Then, applying 
\eqref{eq:toShowCor} with $l=r$, $a=|\alpha|$, $c'=c_{p}$,
$\kappa=\kappa^{\alpha}$  and  $g=g_c$, it follows that there exists a set
$\{w^{\alpha}_j\}_{j=1}^r\in V_r$ such that
\begin{align}
  \left|D^{\alpha}f(c)-\frac{\sum_{j=1}^r w^{\alpha}_j
        g_c(c'_j)}{k^{-|\alpha|}}\right|
  &=\left|g_c^{(|\alpha|)}(c_{p})-\frac{\sum_{j=1}^r w^{\alpha}_j
    g_c(c_{j_\alpha}+\kappa_j/k)}{k^{-a}}\right| \nonumber\\
  &\leq k^{-(r-|\alpha|)}\|g_c\|_r \tilde{C}_r \nonumber\\
  &\leq k^{-(r-|\alpha|)}\|f\|_r \tilde{C}_r.  \label{eq:toShowCor2}
\end{align}

Then, since $c'_j\in\Ck^{(1)}$  for all $j\in\{1,\dots,r\}$ if follows that
there exist a set $\{c^{(j)}\}_{j=1}^r\in\Ck$ such that $ g_c(c'_j)=f(c^{(j)})$
for all $j\in\{1,\dots,r\}$. Noting that $r=\prod_{i=1}^{|\alpha|_0} (r-i+1)$
if $|\alpha|_0=1$, the conclusion of the lemma    holds  with
$C_{|\alpha|,s}=\tilde{C}_r$  for an $\alpha$ such that  $|\alpha|_0=1$.

We now let $\alpha$ be such that $|\alpha|_0\geq 2$ and
$\alpha'\in\mathbb{N}_0^s$ be such that $|\alpha'|_0=|\alpha|_0-1$ and such
that there exists a unique $p\in\{1,\dots,s\}$ for which
$\alpha'_j=\alpha_j$ for all $j\neq p$. Let  $c=(c_1,\dots,c_{s})\in\Ck$
and  $g_c\in \mathcal{C}^{r-|\alpha|'}([0,1])$ be defined by (with obvious
convention when $p\in\{1,s\}$)
\begin{equation*}
g_c(c')=D^{\alpha'}f\left(c_1,\dots,c_{p-1}, c', c_{p+1},\dots, c_s
\right),\quad c'\in[0,1].
\end{equation*}
Note that  $|\alpha|=|\alpha'|+\alpha_{p}$, and thus $D^\alpha
f(c)=g^{(\alpha_{p})}(c_{p})$. 

We now let $\{c'_j\}_{j=1}^{r-|\alpha'|}$ be $r-|\alpha'|$ distinct elements of
the set $\Ck^{(1)}$ such that $|c_{p}-c'_j|\leq (r-|\alpha'|-1)/k$
for all $j\in\{1,\dots,r-|\alpha'|\}$, and 
$\kappa^{\alpha}_j=k(c_j-c_{p})$ for all $j$. Note that the resulting
vector $\kappa^{\alpha}$ is such that $\kappa^{\alpha}\in S_{r-|\alpha'|}$ and
let $\{c^{(j)}\}_{j=1}^{r-|\alpha'|}\subset\Ck$ be such that  (with obvious
convention when $p\in\{1,s\}$)
\begin{equation*}
c^{(j)}=(c_1,\dots,c_{p-1}, c'_j, c_{p+1},\dots, c_s),\quad\forall j\in\{1,\dots,r-|\alpha'|\}.
\end{equation*}

Then, applying  \eqref{eq:toShowCor} with $l=r-|\alpha'|$,
$a=\alpha_{p}$, $c'=c_{p}$, $\kappa=\kappa^{(\alpha)}$  and 
$g=g_c$, it follows that there exists a set
$\{w^{(p)}_j\}_{j=1}^{r-|\alpha'|}\in V_{r-|\alpha'|}$ such that
\begin{equation}\label{eq:ind_1}
\begin{split}
\left|D^{\alpha}f(c)-\frac{\sum_{j=1}^{r-|\alpha'|} w^{(p)}_j
D^{\alpha'}f(c^{(j)})}{k^{-\alpha_{p}}}
\right|&=\left|D^{\alpha}f(c)-\frac{\sum_{j=1}^{r-|\alpha'|} w^{(p)}_j
g_c(c'_j)}{k^{-\alpha_{p}}}  \right|\\
&\leq k^{-(r-|\alpha'|-\alpha_{p})} \|g_c\|_{r-|\alpha'|} \tilde{C}_{r-|\alpha'|}\\
&\leq k^{-(r-|\alpha'|-\alpha_{p})} \|f\|_{r} \tilde{C}_{r-|\alpha'|}\\
&= k^{-(r-|\alpha|)} \|f\|_{r} \tilde{C}_{r-|\alpha'|}.
\end{split}
\end{equation}

To proceed further for $j\in\{1,\dots,r-|\alpha'|\}$ let
\begin{equation*}
    \widehat{D}^{\alpha'}f(c^{(j)})
=k^{|\alpha'|}\sum_{q=1}^{l_{r, \alpha' }} w^{(j)}_{q} f(c^{(j,q)})
\end{equation*}
where $\{ w^{(j)}_{q}\}_{q=1}^{l_{r,\alpha'}}$ and
$\{c^{(j,q)}\}_{q=1}^{l_{r,\alpha'}}$ verify the conditions of the lemma for
$c=c^{(j)}$ and are such that
\begin{equation}\label{eq:ind_2}
    \left|\widehat{D}^{\alpha'}f(c^{(j)})-D^{\alpha'}f(c^{(j)})\right|
\leq k^{-(r-|\alpha'|)} \|f\|_r C_{|\alpha'|,s}
\end{equation}
for some constant $C_{|\alpha'|,s}<\infty$. By the induction hypothesis, there exist sets 
$\{w^{(j)}_{q}\}_{q=1}^{l_{r,\alpha'}}$ and
$\{c^{(j,q)}\}_{q=1}^{l_{r,\alpha'}}$ that verify these conditions. 

We now let
\begin{equation*}
\widehat{D}^{\alpha}_{f}(c)=
k^{\alpha_p}\sum_{j=1}^{r-|\alpha'|} w_{j}^{(p)} \widehat{D}^{(\alpha')}_{f}(c^{(j)})
\end{equation*}
and remark that
\begin{align*}
\widehat{D}^{\alpha}_{f}(c)
& =k^{\alpha_p+|\alpha'|}\sum_{j=1}^{r-|\alpha'|} w_{j}^{(p)}
    \sum_{q=1}^{l_{r, \alpha' }} w^{(j)}_{q} f(c^{(j,q)}) \\ 
&=k^{|\alpha|}\sum_{j=1}^{(r-|\alpha'|)l_{r,\alpha'}}\tilde{w}_j f(\tilde{c}_j)=k^{|\alpha|}\sum_{j=1}^{l_{r,\alpha}}\tilde{w}_j f(\tilde{c}_j)
\end{align*}
where the last equality uses the fact that
\begin{equation*}
(r-|\alpha'|)l_{r,\alpha'}=(r-|\alpha'|)\prod_{i=1}^{|\alpha'|_0} (r-i+1)=\prod_{i=1}^{|\alpha|_0} (r-i+1) 
\end{equation*}
while the penultimate equality holds for a suitable definition of  $\{\tilde{w}_j\}_{j=1}^{l_{r,\alpha}}$ and of $\{\tilde{c}_j\}_{j=1}^{l_{r,\alpha}}$. 

Under the induction hypothesis, each $w^{(j)}_{q}$ is the product of
$|\alpha'|_0$ elements of $\mathcal{W}_r$, and   since each $w_{j}^{p}$
belongs to this set it follows that  each $\tilde{w}_j$ is the product of
$|\alpha'|_0+1=|\alpha|_0$ elements of $\mathcal{K}_r$, as required. It is also
clear that, under the induction hypothesis and the conditions on
$\{c_j\}_{j=1}^{r-|\alpha'|}$ imposed above, the set
$\{\tilde{c}_j\}_{j=1}^{l_{r,\alpha}}$ verifies the assumption of the lemma.

Finally, using \eqref{eq:ind_1} and \eqref{eq:ind_2}, we have
\begin{align*}
\left|D^{\alpha}f(c)-\frac{\sum_{j=1}^{r-|\alpha'|} w_{j}^{p}
\widehat{D}^{(\alpha')}_{f(c_j)}}{k^{-\alpha_{p}}}\right| \leq 
&\left|D^{\alpha}f(c)-\frac{\sum_{j=1}^{r-|\alpha'|} w^{(p)}_j
D^{\alpha'}f(c^{(j)})}{k^{-\alpha_{p}}}  \right|\\
&+\sum_{j=1}^{r-|\alpha'|}
\frac{w^{p}_{j}}{k^{-\alpha_{p}}}\left|\widehat{D}^{(\alpha')}_{f(c_j)}-D^{\alpha'}(c_j)\right|\\
\leq & k^{-(r-|\alpha|)} \|f\|_r \tilde{C}_{r-|\alpha'|}\\
& +k^{\alpha_{p}}\left(\sum_{j=1}^{r-|\alpha'|}|w^{p}_j|\right)  k^{-(r-|\alpha'|)}  \|f\|_r C_{|\alpha'|,s}\\
\leq & k^{-(r-|\alpha|)}\, \|f\|_r
\left(\tilde{C}_{r-|\alpha'|}+\tilde{C}_{r-|\alpha'|}C_{|\alpha'|,s}\right)\\
\leq & k^{-(r-|\alpha|)}\, \|f\|_r C_{|\alpha|,s}
\end{align*}
with $C_{|\alpha|,s}=\tilde{C}_{r-|\alpha'|}(1+C_{|\alpha'|,s})$. The proof is complete.

\subsection{Proof of Lemma~\ref{lemma:main}}\label{p-lemma:main}

Below we only prove the lemma for $\widetilde{\mathcal{I}}_{r,k}(f)$, the proof
for $\widehat{\mathcal{I}}_{r,k}(f)$ being identical.

Let $k\geq r$, $c\in\Ck$, and $h_{k,c}:[-1/2k,1/2k]^s\rightarrow\R$ 
be defined as
$h_{k,c}(u) \eqdef\bar{h}_{k,c}(u)-\E[\bar{h}_{k,c}(U_c)]$ where 
\begin{equation*}
    \bar{h}_{k,c}(u) \eqdef  f(c+u)-\sum_{l=1}^{r-1} \sum_{|\alpha|=l}
    \frac{\widehat{D}^{\alpha}_{k}f(c)}{\alpha!}\left( u^\alpha-
    \prod_{j=1}^s d_{k}(\alpha_j)\right).
\end{equation*}
Then,   $\widetilde{\mathcal{I}}_{r,k}(f)=k^{-s}\sum_{c\in\Ck}  \bar{h}_{k,c}(U_c)$ 
and, since $\E[\widetilde{\mathcal{I}}_{r,k}(f)]=\mathcal{I}(f)$,  we have 
\begin{equation*}
\Itrk{r}(f)-\mathcal{I}(f) = \frac{1}{k^s}\sum_{c\in\Ck}h_{k,c}(U_c).
\end{equation*}

To prove the second part of the lemma let $k\geq r$, $c\in  \Ck$ and
$u\in(-1/2k,1/2k)^s$. Then, using \eqref{eq:Taylor} and with $R_{f,r}$ as in
\eqref{eq:Taylor2},
\begin{align*}
    f(c+u)&=f(c)+ \sum_{l=1}^{r-1} \sum_{|\alpha|=l}\frac{D^\alpha f(c)}{\alpha !}  
 u^\alpha+ R_{f,r}(c,u) 
\end{align*}
so that
\begin{equation}\label{eq:bar_xi}
\begin{split}
h_{k,c}(u) &=
  \sum_{l=1}^{r-1}\sum_{|\alpha|= l} \frac{D^\alpha
f(c)-\widehat{D}^{\alpha}_{k}f(c)}{\alpha!} \left(U^\alpha_c-\prod_{j=1}^s
d_{k}(\alpha_j)\right) \\ 
& +R_{f,r}(c,u) -\E[R_{f,r}(c,U_c)].
\end{split}
\end{equation}
To proceed further, remark that for all $\alpha$ we have
\begin{equation}\label{eq:U1}
u^\alpha\leq (2k)^{-|\alpha|}
\end{equation}
and thus
\begin{equation}\label{eq:R1}
|R_{f,r}(c,u)|\leq k^{-r}  2^{-r}\|f\|_r \sum_{\alpha:|\alpha|=r}\frac{1}{\alpha !}.
\end{equation}
In addition, using \eqref{eq:U1} and noting that $d_{k}(j)\leq k^{-j}$ for all
$j\in\mathbb{N}$, we have
\begin{align}\label{eq:P1}
\left|u^\alpha-\prod_{j=1}^s d_{k}(\alpha_j)\right|
  \leq  (2k)^{-|\alpha|}+k^{-|\alpha|}=k^{-|\alpha|}(2^{-|\alpha|}+1)
\end{align}
while, letting $\bar{C}_{r,s}=\max_{j\in\{1,\dots,r-1\}}C_{j,s}$ with $\{C_{j,s}\}_{j=1}^{r-1}$ as in Lemma \ref{lemma:num_deriv2},
\begin{align}\label{eq:D1}
|D^\alpha f(c)-\widehat{D}^{\alpha}_{k}f(c)|\leq C_{r,s}\|f\|_r
k^{-(r-|\alpha|)}.
\end{align}
Therefore, using \eqref{eq:bar_xi} and \eqref{eq:R1}-\eqref{eq:D1}, it follows that
\begin{align}\label{eq:bound_xi}
|h_{k,c}(u)|\leq \widehat{C}_{s,r}\|f\|_r k^{-r},\quad\forall c\in\Ck,\quad\forall u\in(-1/2k,1/2k)^s
\end{align}
where
\begin{align*}
\widehat{C}_{s,r}= 2 C_{r,s} \sum_{l=1}^{r-1} \sum_{|\alpha|=l}\frac{1}{\alpha !}+2^{-r+1}\sum_{\alpha:|\alpha|=r}\frac{1}{\alpha !}.
\end{align*}
The proof is complete.

\subsection{Proof of Lemma \ref{lemma:var}\label{app:p-thm:clt}}

We prove the result for the estimator  $\widehat{\mathcal{I}}_{r,k}(f)$,
the proof for $\widetilde{\mathcal{I}}_{r,k}(f)$ being identical.

Recall that, for $[a,b]\subset[0,1]^s$ and $f:[0,1]^s\rightarrow\R$, function 
$f_{[a,b]}:[0,1]^s\rightarrow\R$ is defined as 
\begin{equation*}
f_{[a,b]}(u) \eqdef f(a+u(b-a)),\quad u\in [0,1]^s
\end{equation*}
where  the product $u(b-a)$ must be understood as being component-wise. 

We assume without loss of generality that the elements of the set 
$\{\tilde{B}_{q}\}_{q=1}^{p_{r,k}}$ are labelled so that
$\int_{\tilde{B}_q\cap \tilde{ B}_{q'}}\dd u=0$ whenever 
$q,q' \leq \lfloor k/r\rfloor^s$. 
(Recall that the number of $\tilde{B}_q$ is $p_{r,k}=\lceil k/r\rceil^s >
\lfloor k/r \rfloor^s$ when $k/r\notin \mathbb{N}$.)

Then letting
\begin{equation*}
    E_{r,k}=[0,1]^s\setminus \cup_{q=1}^{\lfloor k/ r \rfloor^s}\tilde{B}_q,
\end{equation*}
it follows that
\begin{align*}
  \widehat{\mathcal{I}}_{r,k}(f)\dist\frac{r^s}{k^s}\sum_{q=1}^{\lfloor k/r\rfloor^s}  
  \widehat{\mathcal{I}}_{r,r}(f_{\tilde{B}_q})
+ \widehat{\mathcal{I}}_{r,k}(f\ind_{E_{r,k}}).
\end{align*}
Since these $\lfloor k / r \rfloor^s + 1$ terms are independent, we have 
\begin{equation}\label{eq:var_lower}
\begin{split}
\var(\widehat{\mathcal{I}}_{r,k}(f))&=\frac{r^{2s}}{k^{2s}}\sum_{q=1}^{\lfloor
k/r\rfloor^s}\var\left(\widehat{\mathcal{I}}_{r,r}(f_{\tilde{B}_q})\right)+
\var\left(\widehat{\mathcal{I}}_{r,k}(f\ind_{E_{r,k}})\right).
\end{split}
\end{equation}

We now let $q\in\{1,\dots,\lfloor k/r\rfloor^s\}$ and follow the same lines 
as in \cite[Theorem 2]{haber1969stochastic_bis} in order to compute
$\lim_{p\rightarrow\infty}
\var\left(\widehat{\mathcal{I}}_{r,r}(f_{\tilde{B}_q})\right)$.

To this aim let $\tilde{c}_q$ denote the centre of $\tilde{B}_q$ so that, using
Taylor's theorem, \eqref{eq:Taylor}, we have for all $u\in [0,1]^s$
\begin{align}\label{eq:f_tay}
 f(u)= \sum_{l=0}^{r} \sum_{\alpha:|\alpha|=l}\frac{D^\alpha f(\tilde{c}_q)}{\alpha !}
 ( u-\tilde{c}_q)^{\alpha} + R_{f,r}(\tilde{c}_q,u)
\end{align}
where the function  $R_{f,r}$ is such that \cite[Theorem 5.11, p. 187]{amann2005analysis}
\begin{align}\label{eq:Taylor3}
\lim_{\delta\searrow 0}\,\delta^{-r} \sup_{u,v\in[0,1]^s\,:\|u-v\|\leq\delta}|R_{f,r}(u,v-u)|=0.
\end{align}

Next, let $g:[0,1]^s\rightarrow\R$ be defined by
\begin{equation*}
 g(u) \eqdef \sum_{\alpha: |\alpha|=r}\frac{D^\alpha f(\tilde{c}_q)}{\alpha !}
(u- \tilde{c}_q)^{\alpha}+  R_{f,r}\left(\tilde{c}_q,u\right),\quad   u\in [0,1]^s 
\end{equation*}
and $h \eqdef f-g$.  By Theorem~\ref{thm:main},  
$\widehat{\mathcal{I}}_{r,r}(h_{\tilde{B}_q})=v_q$ a.s. with 
$v_q \eqdef \int h_{\tilde{B}_q}(u)\dd u$ 
since $h$ is a polynomial of degree at most $r-1$.
Hence, using~\eqref{eq:f_tay}, we have
\begin{align*}
    \widehat{\mathcal{I}}_{r,r}(f_{\tilde{B}_q}) - v_q
& = \widehat{\mathcal{I}}_{r,r}(g_{\tilde{B}_q})\\
& =\sum_{\alpha: |\alpha|=r}\frac{D^\alpha f(\tilde{c}_q)}{\alpha !}\,\widehat{\mathcal{I}}_{r,r}
\left(\{(\cdot-\tilde{c}_q)^{\alpha}\}_{\tilde{B}_q}\right)
+ \widehat{\mathcal{I}}_{r, r}\left(R_{f,r}(\tilde{c}_q,\cdot)|_{\tilde{B}_q} \right)\\
& = \frac{r^r}{k^{r}}\sum_{\alpha: |\alpha|=r}\frac{D^\alpha f(\tilde{c}_q)}{\alpha !}
\widehat{\mathcal{I}}_{r,r}\left( (\cdot-1/2)^{\alpha} \right)
+ \widehat{\mathcal{I}}_{r,r}\left(r_{k,q}
\right)
 \end{align*} 
where the function $r_{k,q}$ is defined as 
$r_{k,q}(u) \eqdef R_{f,r}(\tilde{c}_q,\tilde{c}_q-r/(2k)+ur/k)$ for $u\in[0,1]^s$.

This implies that
\begin{align*}
 \E\big[\widehat{\mathcal{I}}_{r,r}(g_{\tilde{B}_q})\big]
& =  \frac{r^r}{k^{r}}\sum_{\alpha: |\alpha|=r}\frac{D^\alpha f(\tilde{c}_q)}{\alpha !} 
  \int_{[0,1]^s} (u-1/2)^{\alpha}\dd u+\int_{[0,1]^s}r_{k,q}(u)\dd u
\end{align*}
and letting
\begin{align*}
    M_\alpha & \eqdef \widehat{\mathcal{I}}_{r,r}\left( (\cdot-1/2)^{\alpha}\right)
-\int_{[0,1]^s} (u-1/2)^{\alpha}\dd u \\
    R_{k, q} & \eqdef \Ih_{r,r}(r_{k, q}) - \int_{[0, 1]^s} r_{k, q}(u)du
\end{align*}
 for all $\alpha$ such that $|\alpha|=r$, we have
 \begin{align*}
 \var\left(\widehat{\mathcal{I}}_{r,r}(f_{\tilde{B}_q})\right) = 
 &\E\left[\left\{\frac{r^{r}}{k^{r}}\sum_{\alpha: |\alpha|=r}\frac{D^\alpha
 f(c)}{\alpha !}M_\alpha+ R_{k, q}] \right\}^2\right]\\
 = &\frac{r^{2r}}{k^{2r}}  \sum_{\alpha,\,\alpha': |\alpha|=|\alpha'|=r}
 \frac{D^\alpha f(c)}{\alpha !}\frac{D^\alpha f(c)}{\alpha' !}
 \E[M_\alpha M_{\alpha'}]+ \var(R_{k, q})\\
& + \frac{2 r^{r}}{k^{r}}\sum_{\alpha: |\alpha|=r}\frac{D^\alpha f(c)}{\alpha
!}\,\E\left[M_\alpha R_{k, q} \right].
 \end{align*}

The above computations show that
\begin{equation}\label{eq:var_terms}
\begin{split}
k^{s+2r}\,&\left(\frac{r^{2s}}{k^{2s}}\sum_{q=1}^{\lfloor k/r\rfloor^s} \var\left( \widehat{\mathcal{I}}_{r,r}(f_{\tilde{B}_q})
\right)\right)\\
= & r^{2r+s}  \sum_{\alpha,\,\alpha': |\alpha|=|\alpha'|=r}\E[M_\alpha
M_{\alpha'}]\left(\frac{r^s}{k^s}\sum_{q=1}^{\lfloor
    k/r\rfloor^s}\frac{D^\alpha f(\tilde{c}_q)}{\alpha !}\frac{D^\alpha
    f(\tilde{c}_q)}{\alpha' !}\right)\\
  &+   r^{2s} k^{2r-s}\sum_{q=1}^{\lfloor k/r\rfloor^s} \var(R_{k,q})\\
&+ 2  r^{r+2s}  k^{r-s}\sum_{q=1}^{\lfloor k/r\rfloor^s}\sum_{\alpha:
|\alpha|=r}\frac{D^\alpha f(c)}{\alpha !}\,\E\left[M_\alpha R_{k, q} \right]
\end{split}
\end{equation}
and we now study in turn each of these three terms.

To study the first term  recall that $B_k(c)$ denotes the hypercube of volume
$k^{-s}$ and centre $c\in\Ck$, and let  $\{c_j\}_{j=1}^{k^s-\lfloor
k/r\rfloor ^s r^s}$ be   the $k^s-\lfloor k/r\rfloor ^s r^s$ elements of $\Ck$
such that     
\begin{equation*}
\int_{B_k(c_j)\cap \tilde{B}_q}\dd u=0,
\quad\forall j\in\{1,\dots,k^s-\lfloor k/r\rfloor ^s r^s\},
\quad\forall q\in\{1,\dots,\lfloor k/r\rfloor ^s\}
\end{equation*}
and such that
\begin{equation*}
\left(\bigcup_{j=1}^{k^s-\lfloor k/r\rfloor ^s r^s}
    B_k(c_j)\right)\bigcup\left(\bigcup_{q=1}^{\lfloor k/r\rfloor ^s
        }\tilde{B}_q\right)=[0,1]^s,
\end{equation*}
and let $\alpha$ and $\alpha'$ be such that $|\alpha|=|\alpha'|=r$. Then, since 
\begin{align*}
\limsup_{k\rightarrow\infty}\left|\frac{1}{k^s}\sum_{j=1}^{k^s-\lfloor
k/r\rfloor ^s r^s} D^\alpha f(c_j)\, D^{\alpha'} f(c_j)\right|&\leq  \limsup_{k\rightarrow\infty} \frac{k^s-\lfloor k/r\rfloor ^s r^s}{k^s}\|f\|_r\\
&\leq \|f\|_r \limsup_{k\rightarrow\infty}  \left(1-(1-r/k)^s
\right)=0
\end{align*}
and because the Riemann sum
\begin{equation*}
    \frac{r^s}{k^s}\sum_{q=1}^{\lfloor k/r\rfloor^s} D^\alpha
    f(\tilde{c}_q)\,D^{\alpha'} f(\tilde{c}_q)
 +\frac{1}{k^s}\sum_{j=1}^{k^s-\lfloor k/r\rfloor ^s r^s} D^\alpha f(c_j)\,D^{\alpha'} f(c_j)
\end{equation*}
converges to  $\int_{[0,1]^s}D^\alpha f(u)\,D^{\alpha'} f(u)\dd u$ as $k\rightarrow\infty$, it follows that
\begin{align}\label{eq:CLT_1}
\lim_{k\rightarrow\infty} \left\{\frac{r^s}{k^s}
\sum_{q=1}^{\lfloor k/r\rfloor^s} D^\alpha f(\tilde{c}_q)\,D^{\alpha'} f(\tilde{c}_q) \right\}
=\int_{[0,1]^s}D^\alpha f(u)\,D^{\alpha'} f(u)\dd u.
\end{align}
 
Next, using \eqref{eq:Taylor3} we have
\begin{align}\label{eq:var_terms2}
\limsup_{k\rightarrow\infty}\left\{ k^{2r-s} \sum_{q=1}^{\lfloor
k/r\rfloor^s}\var(R_{k, q})\right\}
& \leq r^{-s} \times \limsup_{k\rightarrow\infty} \left\{ k^{2r} 
    \max_{1\leq q \leq \lfloor  k/r\rfloor^s} \E[R_{k,q}^2]\right\} \nonumber \\
& =0.
\end{align}

Finally, noting that for some constant $C<\infty$ we have, $\P$-a.s., $|M_\alpha|\leq C$ for all $\alpha$ such that $|\alpha|=r$, it follows that
\begin{multline}\label{eq:var_terms3}
\limsup_{k\rightarrow\infty} \left|k^{r-s} \sum_{q=1}^{\lfloor
k/r\rfloor^s}\sum_{\alpha: |\alpha|=r}  \frac{D^\alpha f(\tilde{c}_q)}{\alpha
!}\, \E\left[M_\alpha R_{k,q}\right]\right|\\
\leq 2 C  r^{-s}\|f\|_r   \left(\sum_{\alpha: |\alpha|=r} \frac{1}{\alpha !}\right)
\limsup_{k\rightarrow\infty} \left\{ k^r
    \max_{1 \leq q\leq \lfloor k/r\rfloor^s} \E[|R_{k,q}|]\right\} = 0
\end{multline}
where the equality holds by \eqref{eq:Taylor3}.

Therefore, combining \eqref{eq:var_terms}-\eqref{eq:var_terms3}, we obtain
\begin{multline}\label{eq:part1}
\lim_{k\rightarrow\infty} k^{s+2r}\,\left\{\frac{r^{2s}}{k^{2s}}\sum_{q=1}^{\lfloor k/r\rfloor^s}
\var\left( \widehat{\mathcal{I}}_{r,r}(f_{\tilde{B}_q})
\right)\right\}\\
= r^{2r+s}\sum_{\alpha,\,\alpha': |\alpha|=|\alpha'|=r}\frac{\E[M_\alpha M_{\alpha'}]}{\alpha!\alpha'!}\int_{[0,1]^s}D^\alpha f(u)\,D^{\alpha'} f(u)\dd u 
\end{multline}
and thus, by \eqref{eq:var_lower}, to conclude the proof of the lemma it remains to show that
\begin{equation}\label{eq:part2}
\lim_{k\rightarrow\infty} \left\{
k^{s+2r}\var\left(\widehat{\mathcal{I}}_{r,k}(f\ind_{E_{r,k}})\right)\right\}=0.
\end{equation}

To this aim let $m_k \eqdef p_{r,k}-\lfloor k/r\rfloor^s$,
$\{c_j\}_{j=1}^{m_k}$ be such that 
$\cup_{j=1}^{m_k}B_k(c_j)=E_{r,k}$
and note that
\begin{align*}
\frac{k^s}{m_k} \widehat{\mathcal{I}}_{r,k}(f\ind_{E_{r,k}})
= & \frac{1}{m_k}\sum_{j=1}^{m_k}\frac{f(c_j+U_{c_j})+f(c_j-U_{c_j})}{2}\\
&- \frac{1}{m_k}\sum_{j=1}^{m_k} \sum_{l=1}^{\lfloor
(r-1)/2\rfloor}\sum_{\alpha:\, |\alpha|=2l}\frac{\widehat{D}^{\alpha}_{k,
f(c_j)}}{\alpha!}\left( U^\alpha_{c_j}- \prod_{j=1}^s d_{k}(\alpha_j)\right).
\end{align*}

Then, using Lemma~\ref{lemma:main}  
\begin{align*}
\var\left((k^s/m_k)\widehat{\mathcal{I}}_{r,k}(f\ind_{E_{r,k}})\right)
& \leq \widehat{C}^2_{s,r}\|f\|^2_r m_k^{-1} k^{-2r}\\
& \Leftrightarrow \var\left( \widehat{\mathcal{I}}_{r,r}(f_{E_{r, k}}\right)
\leq m_k k^{-2s-2r}\widehat{C}^2_{s,r}\|f\|^2_r
\end{align*}
where $\widehat{C}_{s,r}<\infty$ is as in  Lemma \ref{lemma:main}.

Therefore, noting that
\begin{equation*}
m_k=\lceil k/r\rceil^s-\lfloor k/r\rfloor^s\leq k^{s}
\left\{ (r^{-1}+k^{-1})^s-(r^{-1}-k^{-1})^s\right\},
\end{equation*}
we have
\begin{multline*}
\limsup_{k\rightarrow\infty}  
\left\{k^{2r+s} \var\left( \widehat{\mathcal{I}}_{r,r}(f_{E_{r,k}})
\right)\right\} \\
\leq \limsup_{k\rightarrow\infty}
\left\{ (r^{-1}+k^{-1})^s-(r^{-1}-k^{-1})^s\right\} \widehat{C}^2_{s,r}\|f\|^2_r
=0.
\end{multline*}

This shows \eqref{eq:part2} and the proof of the lemma is complete.

\subsection{Proof of Lemma~\ref{lemma:expectation}}


Recall that $B_k(c)$ denotes the hyper-cube $[c-1/2k, c+1/2k] = \prod_{i=1}^s
[c_i-1/2k, c_i+1/2k]$, with centre $c$ and volume $k^{-s}$. 
Treating $k$ as fixed from now on, we define for $j\in\mathbb{N}_0$,
$\mathcal{B}_{j,1} = \{ B_k(c) \}_{c\in\Cx{j, k}}$, and, for $l=3, 5, \ldots$, 
we define $\mathcal{B}_{j,l}$ to be the set of hyper-cubes $B_{k/l}(c)$, which
are then unions of $l^s$ elements in $\mathcal{B}_{j, 1}$.
We also treat as fixed $\lambda \in \{\pm (2i+1),\,i\in\mathbb{N}_0\}$,
$p=(|\lambda|-1)$ and $m\geq p/2$. 


Consider a given $c\in\Cmk$.  We have
$[c-\lambda/2k,c+\lambda/2k]\in\mathcal{B}_{m,|\lambda|}$ and thus  there exist
distinct hypercubes
$\{B_{c,l}\}_{l=1}^{|\lambda|^s}$ in $\mathcal{B}_{m, 1}$ such that
\begin{equation*}
[c-\lambda/2k,c+\lambda/2k]=\bigcup_{l=1}^{|\lambda|^s} B_{c,l}.
\end{equation*}

For $U_c\sim\mathcal{U}([-1/2k,1/2k]^s)$, $\bar{g}$ defined as in the statement of the
lemma, we have
\begin{align}
\E[\bar{g}(c+\lambda U_c)]
&=k^s\int_{[-1/2k,1/2k]^s}\bar{g}(c+\lambda u)\dd u \nonumber \\
&=\frac{k^s}{|\lambda|^s}\int_{[0,1]^s\cap  [c-\lambda/2k,c+\lambda/2k]}g(u)\dd u
\nonumber \\
&=\frac{k^s}{|\lambda|^s}\sum_{l=1}^{|\lambda|^s}
\int_{[0,1]^s\cap B_{c,l}} g(u)\dd u.  \label{eq:exp}
\end{align}

To proceed further we remark that (again, recall $m\geq p/2$, otherwise this
would not be true):
\begin{equation*}
\bigcup_{c\in\Cmk }\left\{[0,1]^s\cap \cup_{l=1}^{|\lambda|^s}B_{c,l}\right\}
= \underbrace{\mathcal{B}_{0,1}\cup\dots\cup \mathcal{B}_{0,1}}_{\text{$|\lambda|^s$ times}} 
\end{equation*}
which, together with~\eqref{eq:exp}, yields
\begin{align*}
\frac{1}{k^{s}}\sum_{c\in\mathfrak{C}_{m,k}}\E[\bar{g}(c+\lambda U_c)] 
& =\frac{1}{|\lambda|^s}\sum_{c\in\Cmk}\sum_{l=1}^{|\lambda|^s}\int_{[0,1]^s\cap
B_{c,l}}g(u)\dd u \\
&=\sum_{B\in\mathcal{B}_{0,1}}\int_{B}g(u)\dd u\\
&= \int_{[0,1]^s}g(u)\dd u.
\end{align*}
The proof is complete.
\end{document}